\documentclass[a4paper,10pt]{article}
\pdfoutput=1 

\usepackage{jheppub} 
\usepackage{amssymb,amstext,amsthm,amsfonts,mathrsfs,amsbsy,amsmath,array,multirow}
\usepackage[T1]{fontenc} 
\newtheorem{thm}{Theorem}[section]

\newtheorem{prop}[thm]{Proposition}
\newtheorem{lemma}[thm]{Lemma}
\newtheorem{cor}[thm]{Corollary}

\newtheorem{definition}[thm]{Definition}

\newtheorem{remark}[thm]{Remark}

\newcommand{\beq}{\begin{equation}\begin{aligned}}
\newcommand{\eeq}{\end{aligned}\end{equation}}

\newcommand{\ii}{\mathrm{i}}

\rightline{\small IPhT-t14/056}

\title{\boldmath Spin(7)-manifolds in compactifications to four dimensions}

\author[a]{Mariana Gra\~na,}
\author[a]{C. S. Shahbazi}
\author[b]{and Marco Zambon}

\affiliation[a]{Institut de Physique Th\'eorique, CEA Saclay.}
\affiliation[b]{Universidad Aut\'onoma de Madrid and ICMAT(CSIC-UAM-UC3M-UCM), Madrid.}

\emailAdd{mariana.grana@cea.fr}
\emailAdd{carlos.shabazi@cea.fr}
\emailAdd{ marco.zambon@icmat.es, marco.zambon@uam.es}

\abstract{\vspace{0.1cm}\\ We describe off-shell $\mathcal{N}=1$  M-theory compactifications down to four dimensions in terms of eight-dimensional manifolds equipped with a topological $Spin(7)$-structure. Motivated by the exceptionally generalized geometry formulation of M-theory compactifications, we consider an eight-dimensional manifold $\mathcal{M}_{8}$ equipped with a particular set of tensors $\mathfrak{S}$ that  allow to naturally embed in $\mathcal{M}_{8}$ a family of $G_{2}$-structure seven-dimensional manifolds as the leaves of a codimension-one foliation. Under a different set of assumptions, $\mathfrak{S}$ allows to make $\mathcal{M}_{8}$ into a principal $S^{1}$ bundle, which is equipped with a topological $Spin(7)$-structure if the base is equipped with a topological $G_{2}$-structure. We also show that $\mathfrak{S}$ can be naturally used to describe regular as well as a singular elliptic fibrations on $\mathcal{M}_{8}$,  which may be relevant for F-theory applications,  
and prove several mathematical results concerning the relation between topological $G_{2}$-structures in seven dimensions and topological $Spin(7)$-structures in eight dimensions. }

\arxivnumber{1402.xxxx}

\begin{document}
\maketitle
\flushbottom



\section{Introduction}
\label{sec:introduction}


Supersymmetric String-Theory compactifications have always been a beautiful source of connections between physics and differential and algebraic geometry\footnote{See for example the reviews \cite{Greene:1996cy,Grana:2005jc,Denef:2008wq,Koerber:2010bx} for more details and further references.}. The supersymmetry equations impose topological as well as differential conditions on the space-time manifold that can be nicely codified using the different tools existing in differential geometry, and in particular, the notion of topological and geometric $G$-structures \cite{Chern}. For instance, in the seminal paper \cite{Candelas:1985en}, it was proven that the internal space of a four-dimensional compactification of flux-less heterotic String Theory must be a Calabi-Yau three-fold, that is, a six-dimensional real manifold of $SU(3)$-holonomy. Likewise, the internal space of a flux-less, $\mathcal{N}=1$ four-dimensional M-Theory compactification, must be a seven-dimensional manifold of $G_{2}$-holonomy \cite{Duff:1986hr}, and the internal manifold of a flux-less, $\mathcal{N}=1$ four-dimensional F-theory compactification, must be a Calabi-Yau four-fold, namely an eight-dimensional manifold of $SU(4)$-holonomy \cite{Vafa:1996xn,Becker:1996gj}. As a general rule, the supersymmetry conditions on a flux-less compactification of String/M/F-Theory imply that the internal space has to be a manifold of special holonomy \cite{Salamon,Joyce2007}. Remarkably enough, mathematicians had started studying manifolds of special holonomy thirty years before their appearance in String Theory \cite{Bonan,Berger}, and the study of these manifolds is nowadays still an active field of research in mathematics. In order to write a lower-dimensional effective theory encoding the dynamics of the massless degrees of freedom of the compactified theory, it is necessary to know the moduli space of the compactification manifold \cite{Duff:1986hr}. For manifolds of special holonomy, the moduli space is relatively well understood; for example, the moduli space of a Calabi-Yau three-fold is itself a K\"ahler manifold that factorizes as the product of two projective Special-K\"ahler manifolds \cite{Candelas:1990pi}. For the case of $G_{2}$-holonomy structures, several results are summarized in reference \cite{Joyce2007}, where the moduli space is given as a finite-dimensional subspace of the infinite-dimensional space of sections of a particular bundle over the seven-dimensional internal manifold. Characterizing the moduli space of geometric structures satisfying a differential condition on a manifold as the finite subspace of the infinite-dimensional space of sections of an appropriate bundle will prove to be the right way to proceed also in more complicated examples, namely in the presence of fluxes.

Compactifications with non-trivial fluxes are considerably more complicated than its flux-less counterparts \cite{Kachru:2002he,Grana:2005jc,Douglas:2006es,Denef:2007pq}. Although the supersymmetry conditions in the presence of fluxes have been mostly worked out, starting from \cite{Strominger:1986uh} for the case of the heterotic string, solving them is extremely complicated. This means that already the task of obtaining, or characterizing in any meaningful way, the vacuum of the compactification may be an impossible task. The first consequence of having non-trivial fluxes is that the internal manifold is not going to be in general a manifold of special holonomy. In addition, the moduli space of manifolds satisfying more cumbersome differential conditions are poorly understood, and therefore the task of obtaining an explicit effective action for the compactification theory in the presence of fluxes becomes much harder. 

Still, lot of effort has been devoted in order to better understand flux compactifications, and remarkable progress has been made. Recently, two mathematical tools have been developed, to wit, generalized complex geometry (GCG)  \cite{2002math......9099H,2004math......1221G} and exceptional generalized geometry (EGG) \cite{Hull:2007zu,Pacheco:2008ps}, that turned out to be very powerful, among their many other applications, in order to study supersymmetric solutions and supersymmetric compactifications of String and M-Theory. The key point in order to apply GCG and EGG to String/M-Theory is to use the fact that supersymmetry conditions on the space-time manifold are sometimes better characterized not using tensors but sections of different vector bundles over the internal compactification space, which in addition imply the corresponding topological reductions on the associated principal bundles. On-shell supersymmetry can then be expressed through differential conditions on the sections of this extended bundles. In the case of GCG this new bundle is the sum of the tangent and the cotangent bundle of the internal space, whereas in EGG it is a more complicated extension of the tangent bundle, such that all the charges in String Theory (including RR ones) or in M-theory are geometrized. 

M-Theory compactifications with (a priori off-shell) ${\mathcal N}=1$ supersymmetry are described using a 912-rank vector bundle $E\to\mathcal{M}_{7}$ where the structure group is $E_{7(7)}$ acting on the ${\bf 912}$ representation $E_{{\bf 912}}$ and $\mathcal{M}_{7}$ denotes the internal seven-dimensional manifold \cite{Pacheco:2008ps}.  
The ${\bf 912}$ representation is naturally decomposed in terms of the fundamental representation of $Sl(8,\mathbb{R})$ acting on an eight-dimensional vector space $V$, as follows
\begin{equation} 
\label{eq:912intro}
E_{{\bf 912}} = S^2 V\oplus \left(\Lambda^{3}V\otimes V^{\ast}\right)_{0}\oplus S^2 V^{\ast}\oplus \left(\Lambda^{3}V^{\ast}\otimes V\right)_{0}\, .
\end{equation}
\noindent
where $S^2$ denotes symmetric two-tensors and the subindex 0 denotes traceless. The eight-dimensional vector space $V$ cannot be straightforwardly identified with the tangent space of the internal compactification manifold, being the latter seven-dimensional. It is then suggestive to try to find an \emph{eight-dimensional structure} encoding the natural decomposition of $E_{{\bf 912}}$ in terms of the fundamental representation of $Sl(8,\mathbb{R})$. In this paper we will pursue this idea, originally proposed in \cite{Grana:2012zn}, by studying an eight-dimensional manifold $\mathcal{M}_{8}$ equipped at every point with a set of tensors $\mathfrak{S}$ defining the decomposition of $E_{{\bf 912}}$ in terms of $Sl(8,\mathbb{R})$, as given in \eqref{eq:912intro}. We will study the geometric properties of $\mathcal{M}_{8}$ relating it to the internal manifolds used in M and F-Theory compactifications. Remarkably enough, the structure $\mathfrak{S}$ prescribed by the decomposition \eqref{eq:912intro}, is precisely appropriate in order to embed in $\mathcal{M}_{8}$ a family of manifolds of $G_{2}$ structure, relevant in M-Theory compactifications, as well as, at the same time, proving $\mathcal{M}_{8}$ to have a topological $Spin(7)$ structure, which allows to relate $\mathfrak{S}$ to supersymmetry. In addition, this gives a particular relation between $G_{2}$-structure seven-dimensional and $Spin(7)$- structure eight-dimensional manifolds, which may be of physical interest in the context of String/M/F-Theory dualities. Notice that in order to conclude that $\mathcal{M}_{8}$ is an admissible internal space in F-Theory, it has to be elliptically fibered. It turns out that $\mathfrak{S}$ is an appropriate structure to define, under some mild assumptions, a regular as well as a singular elliptic fibration in $\mathcal{M}_{8}$.

In this paper we will consider exclusively off-shell supersymmetry and therefore the structures involved will always be topological. Clearly, more effort is needed in order to understand better the relevance of $\mathcal{M}_{8}$ and $\mathfrak{S}$ in String/M/F-Theory, the first step being to consider on-shell supersymmetry and therefore differential conditions on $\mathfrak{S}$. We leave that for a future publication. What we will  unravel  here is the more general question of how to build $Spin(7)$-structure eight-dimensional manifolds from $G_{2}$-structure seven-dimensional manifolds, or, more in general how $G_{2}$-structure and $Spin(7)$-structure manifolds are related. This is a question of very much physical interest, given that $G_{2}$-structure manifolds are important in M-Theory compactifications and $Spin(7)$-structure manifolds are important in F-Theory compactifications. Therefore the link between both structures should have an interpretation in terms of string dualities.

The paper is organized as follows. In section \ref{sec:EGGM} we review M-Theory compactifications preserving some off-shell supersymmetry (focusing in particular on ${\mathcal N}=1$) and give the corresponding exceptional generalized geometric formulation, which is also used to motivate the set of tensors $\mathfrak{S}$ defined on $\mathcal{M}_{8}$. Since it will be important later in order to study the geometry of $\mathcal{M}_{8}$ equipped with $\mathfrak{S}$, in section \ref{sec:G2in7} we give a fairly complete review of $G_{2}$-structures on seven-dimensional vector spaces and $Spin(7)$-structures on eight-dimensional vector spaces, and obtain several results about the relation between them. In section \ref{sec:8manifold} we precisely define $\mathfrak{S}$ on $\mathcal{M}_{8}$ and state the corresponding existence obstructions. In section \ref{sec:nd} we begin the proper study of $\mathcal{M}_{8}$ equipped with $\mathfrak{S}$, obtaining several results concerning $Spin(7)$-structures on $\mathcal{M}_{8}$ and the relation of $\mathcal{M}_{8}$ with seven-dimensional manifolds of $G_{2}$-structure. In section \ref{sec:generalintermediate} we consider the relation of $\mathfrak{S}$ to regular as well as singular elliptic fibrations on $\mathcal{M}_{8}$, in order to evaluate the viability of $\mathcal{M}_{8}$ as an internal space in F-Theory compactifications. We conclude in section \ref{sec:conclusions}. More details on $G_2$ and $Spin(7)$ manifolds can be found respectively in Appendices \ref{sec:G2appendix} and \ref{sec:Spin7appendix}.


\section{M-theory compactifications}
\label{sec:EGGM}


The effective, low-energy, description of M-theory \cite{Witten:1995ex} is believed to be given by eleven-dimensional $\mathcal{N}=1$ Supergravity \cite{Cremmer:1978km}, whose field content is given by a Lorentzian metric $\mathsf{g}$, a three-form gauge field $C$ and a Majorana gravitino $\Psi$. We are interested in bosonic solutions to eleven-dimensional Supergravity, and therefore we will give from the onset a zero vacuum expectation value to the Majorana gravitino $\Psi$. The bosonic action of classical eleven-dimensional Supergravity reads

\begin{equation}
\label{eq:11sugra}
S = \int_{\mathcal{M}}\left\{ \mathsf{R}\, d\mathcal{V} - \frac{1}{4} \mathsf{G}\wedge\ast\,\mathsf{G} + \frac{1}{12} \mathsf{G}\wedge\mathsf{G}\wedge \mathsf{C} \right\} \, ,
\end{equation}

\noindent
where $\mathcal{M}$ denotes the eleven-dimensional space-time differentiable\footnote{By differentiable manifold we mean a Hausdorff, second-countable, topological space equipped with a \emph{differentiable structure}.}, orientable and spinnable\footnote{Since it is spinnable the frame bundle $F(\mathcal{M})\to\mathcal{M}$ of $\mathcal{M}$ admits a spin structure $\tilde{F}(\mathcal{M})$ such that its associated vector bundle $S\to\mathcal{M}$ is the spin bundle over $\mathcal{M}$ manifold. In eleven dimensions with signature $(1,10)$, the spin bundle $S_{p}\, , \,\, p\in\mathcal{M}$ is the thirty-two-dimensional Majorana representation $\Delta^{\mathbb{R}}_{1,10}$ and sections $\epsilon\in\Gamma\left( S\right)$ of $S$ are Majorana spinors.} manifold, $d\mathcal{V}$ is the canonical volume form induced by the metric $\mathsf{g}$, and $G$ is the closed four-form flux associated to $C$, i.e. locally we can write $\mathsf{G} = d\mathsf{C}$. 
In a bosonic background $(\mathsf{g},\mathsf{C})$, the only non-trivial supersymmetry transformation is the gravitino one, given by\footnote{We denote by $\flat$ and $\sharp$ the musical isomorphisms defined by the manifold metric.}

\begin{equation}
\label{eq:susytransformation}
\delta_{\epsilon} \psi(v)= \nabla^{S}_{v} \epsilon + \frac{1}{6} \iota_{v} \mathsf{G}\cdot \epsilon +\frac{1}{12} v^{\flat}\wedge \mathsf{G}\cdot \epsilon \, , \qquad  
\end{equation}

\noindent
where $\psi(v)=\iota_v \psi$, $\epsilon\in\Gamma\left( S\right)$, and $\cdot$ denotes the Clifford multiplication. 

In this letter we study compactifications of M-theory down such that the resulting four dimensional theory is ${\cal N}=1$ supersymmetric off-shell.  Although we speak about \emph{compactifications}, we are not going to assume that $\mathcal{M}_{7}$ is compact for two basic reasons. First, it may be consistent to compactify in non-compact space-times with finite volume and appropriate behavior of the laplacian operator \cite{Nicolai:1984jga}. And secondly, the results that we will obtain in the rest of the paper involve seven-dimensional manifolds that may not be necessarily compact.
We we will assume that the space-time manifold $\mathcal{M}$ can be written as the direct product of four-dimensional Lorentzian  space-time $\mathcal{M}_{1,3}$ and a seven dimensional, Riemannian, orientable and spinnable manifold $\mathcal{M}_{7}$

\begin{equation}
\label{eq:productmanifold}
\mathcal{M} = \mathcal{M}_{1,3}\times\mathcal{M}_{7}\, ,
\end{equation}

\noindent
and we will therefore take the Lorentzian metric $g_{11}$ on $\mathcal{M}$ to be given by

\beq
\label{eq:11dmetricproduct}
g_{11}= g_{1,3}\times g_{7} \, . 
\eeq

\noindent
Given the product structure (\ref{eq:productmanifold}) of the space-time manifold $\mathcal{M}$, the tangent bundle splits as follows\footnote{We omit the pull-backs of the canonical projections.}

\begin{equation}
T\mathcal{M} = T\mathcal{M}_{1,3}\oplus T\mathcal{M}_{7}\, ,
\end{equation}

\noindent
which allows, using (\ref{eq:11dmetricproduct}), a decomposition of the structure group $\mathrm{Spin}(1,10)$ of $\mathcal{M}$ in terms of $\mathrm{Spin}(1,3)\times \mathrm{Spin}(7)\subset \mathrm{Spin}(1,10)$ representations. The corresponding branching rule is

\begin{equation}
\label{eq:branchingrule}
\Delta^{\mathbb{R}}_{1,10} = \Delta^{+}_{1,3} \otimes \left(\Delta^{\mathbb{R}}_{7}\right)_{\mathbb{C}}  \oplus \Delta^{+ \ast}_{1,3} \otimes \left(\Delta^{\mathbb{R}}_{7}\right)^{\ast}_{\mathbb{C}}\, ,
\end{equation}

\noindent
where $\Delta^{+}_{1,3}$ denotes the positive-chirality complex Weyl representation of $\mathrm{Spin}(1,3)$, $\Delta^{\mathbb{R}}_{7}$ denotes the real Majorana representation of $\mathrm{Spin}(7)$ and the subscript $\mathbb{C}$ denotes the complexification of the corresponding real representation. Let us respectively denote by $S^{+}_{1,3}$ and $S^{\mathbb{R}}_{7}$ the corresponding spin bundles over $\mathcal{M}_{1,3}$ and $\mathcal{M}_{7}$. Using equation (\ref{eq:branchingrule}) we deduce that the supersymmetry spinor $\epsilon\in\Gamma\left( S\right)$ decomposes as follows 

\begin{equation}
\label{eq:11dspinor}
\epsilon=\xi_+ \otimes \eta +  \xi_{+}^c \otimes {\eta}^c\, ,
\end{equation}

\noindent
where $\xi_{+}\in \Gamma\left( S^{+}_{1,3}\right)$ and 
\beq
\eta = \eta_{1} + i\eta_{2} \ , \qquad \eta_{1},\eta_{2}\in \Gamma\left( S^{\mathbb{R}}_{7}\right) \ . 
\eeq
\noindent Given the decomposition (\ref{eq:11dspinor}), $\mathcal{M}_{7}$ is equipped with a globally defined no-where vanishing complex spinor $\eta$, that is a globally defined section of the complexified spin bundle $S^{\mathbb{R}}_{7}\otimes\mathbb{C}\to\mathcal{M}_{7}$. Its real components $\eta_{1}$ and $\eta_{2}$ can a priori vanish at points or become parallel, as long as they do not simultaneously vanish.


\subsection{The seven dimensional manifold ${\cal M}_7$}
\label{sec:sevendimmanifolds}


Generically, the existence of globally defined nowhere vanishing spinors implies a topological reduction of the structure group of the frame bundle from $SO(7)$ (or rather Spin(7) for a spin manifold) to a given subgroup. In the case of seven-dimensional spin manifolds, however, the reduction of the structure group is guaranteed due to the following proposition\footnote{$G_2$-structures will be introduced in detail in section \ref{sec:G2in7}. Further definitions and properties are given in  Appendix \ref{sec:G2appendix}.}
\begin{prop}
\label{prop:topologicalreductions7noncompact}
Let $\mathcal{M}$ be a seven-dimensional manifold. Then the following conditions are equivalent

\begin{enumerate}

\item $\mathcal{M}$ admits a topological $\mathrm{Spin}(7)$-structure.

\item The first and the second Stiefel-Whitney class of $\mathcal{M}$ vanish, that is $\omega_{1} = 0$ and $\omega_{2} = 0$.

\item $\mathcal{M}$ admits a topological $G_{2}$-structure.

\end{enumerate}
\end{prop}

\begin{proof}
The equivalence $1\Leftrightarrow 2$ and the implication $3\Rightarrow 1$  are obvious. We have to show therefore that the existence of a topological $\mathrm{Spin}(7)$-structure implies the existence of a topological $G_{2}\subset \mathrm{Spin}(7)$ structure. Let $S^{\mathbb{R}}$ denote the real spin bundle associated to the $\mathrm{Spin}(7)$ structure. Since its real dimension is eight, and the real dimension of $\mathcal{M}$ is seven, there exists a global section $\psi\in\Gamma\left( S^{\mathbb{R}}\right)$ of unit length. Since the $G^{c}_{2}$ can be defined as the isotropy group of a given real spinor of $\mathrm{Spin}(7)$, $\psi$ can be used to define a topological $G^{c}_{2}$ structure on $\mathcal{M}$. 
\end{proof}

On a seven-dimensional spin manifold we have therefore always one globally defined no-where vanishing spinor. 
For compact seven dimensional manifolds, the implications of a spin structure are even stronger, namely

\begin{prop}
\label{prop:topologicalreductions7compact}
Let $\mathcal{M}$ be a compact seven-dimensional manifold. Then the following conditions are equivalent

\begin{enumerate}

\item $\mathcal{M}$ admits a topological $\mathrm{Spin}(7)$-structure.

\item The first and the second Stiefel-Whitney class of $\mathcal{M}$ vanish, that is $\omega_{1} = 0$ and $\omega_{2} = 0$.

\item $\mathcal{M}$ admits a topological $G_{2}$-structure.

\item $\mathcal{M}$ admits a topological $SU(3)$-structure.

\item $\mathcal{M}$ admits a topological $SU(2)$-structure.

\end{enumerate}
\end{prop}

\begin{proof}See proposition 3.2 in \cite{nearlyparallel}.
\end{proof} 

\noindent
Therefore, if $\mathcal{M}_{7}$ is orientable, compact and spin, it admits a topological $SU(2)$ structure, or equivalently three globally defined nowhere vanishing and nowhere parallel spinors. 

In particular, as shown in appendix \ref{sec:G2appendix}, there is a one to one correspondence between $G_{2}$ structures, positive three-forms\footnote{See appendix \ref{sec:G2appendix} for the definition of positive three-form.} and $Spin(7)$ real spinors  on a seven-dimensional manifold. Therefore, \emph{for every} $G_{2}$\emph{-structure} there is automatically a positive three-form and a $Spin(7)$ real spinor globally defined on $\mathcal{M}$. 

For compactifications preserving off-shell supersymmetry, ${\cal M}_7$ is equipped with two spinors $\eta_{1}\, , \eta_{2} \in \Gamma\left(S^{\mathbb{R}}_{7}\right)$, which may allow for a further reduction of the structure group $G_{2}$ of $\mathcal{M}_{7}$. Since every orientable and spin $G_{2}$ manifold already has a globally defined section $\psi\in\Gamma\left(S^{\mathbb{R}}_{7}\right)$ of the spin bundle, we might have three globally defined spinors on $\mathcal{M}_{7}$, namely $\eta_{1}\, , \eta_{2}$ and $\psi$, which globally give rise to three different possibilities

\begin{itemize}

\item The three spinors are nowhere vanishing and nowhere linearly dependent, which implies that they define a topological $SU(2)$-structure on $\mathcal{M}_{7}$.

\item  Only two spinors are nowhere vanishing and linearly independent, which implies that they define a topological $SU(3)$-structure on $\mathcal{M}_{7}$.

\item  The three spinors are linearly dependent everywhere, which implies that they define a topological $G_{2}$-structure on $\mathcal{M}_{7}$, something that is always guaranteed due to proposition \ref{prop:topologicalreductions7noncompact}.

\end{itemize}

Of course, there exists the possibility that the spinors become linearly dependent only at some points in $\mathcal{M}_{7}$. In such situation there is no globally well-defined topological reduction of the frame bundle further than the $G_2$ one, but there is, in the case of two $\mathrm{Spin}(7)$ spinors, a well defined global reduction in the generalized bundle $\mathbb{E} = T\mathcal{M}_{7}\oplus T^{\ast}\mathcal{M}_{7}$ from $\mathbb{R}^{\ast}\times\mathrm{Spin}(7,7)$ to $G_{2}\times G_{2}$ \cite{2005math......2443W,2006CMaPh.265..275W}. 

In the case when the three spinors are nowhere vanishing (no matter what their respective inner product is), due to the isomorphism $S^{\mathbb{R}}_{7}\otimes S^{\mathbb{R}}_{7} \simeq \Lambda^{\bullet} T^{\ast}\mathcal{M}_{7}$ we can write

\begin{equation}
\label{eq:etaforms}
\eta_{a}\otimes\eta_{a} = 1 + \phi_{a} + \ast\phi_{a} + \mathcal{V}_{g_{a}}\, , \qquad a=1,2\, ,
\end{equation}

\begin{equation}
\label{eq:psiforms}
\psi\otimes\psi = 1 + \phi_{\psi} + \ast\phi_{\psi} + \mathcal{V}_{g_{\psi}}\, ,
\end{equation}

\noindent
where $\phi_{a}$ and $\phi_\psi$ are the corresponding positive forms associated to $\eta_a$ and $\psi$ and $\mathcal{V}_{g_{a}}\, , \mathcal{V}_{g_{\psi}}$ are the volume forms associated to the metric defined by the corresponding positive three-form. In terms of spinor bilinears we have

\begin{equation} \label{eq:3forma}
\phi_{a} = \ii \eta^{T}\gamma_{(3)}\eta\, , \quad a=1,2\, ,\qquad \phi_{\psi} = \ii \psi^{T}\gamma_{(3)}\psi\, ,
\end{equation}
where $\gamma_{(3)}$ is the anti-symmetrized product of three gamma matrices.
\noindent
In fact, using $\eta_{a}$ and $\psi$ we can construct many more forms on $\mathcal{M}_{7}$ than those appearing in equations (\ref{eq:etaforms}) and (\ref{eq:psiforms}). They can be used to alternatively define the corresponding reductions of the structure group of $\mathcal{M}_{7}$ \cite{Kaste:2003zd}. 

Backgrounds preserving ${\mathcal N}=1$ supersymmetry (on-shell), should be invariant under a supersymmetry transformation. In the case of bosonic backgrounds, the only non-trivial one is that of the gravitino transformation, given in equation (\ref{eq:susytransformation}). Supersymmetry requires equation (\ref{eq:susytransformation}) to vanish for any vector $v$, and therefore implies differential equations on the supersymmetry spinors $\eta_{1}$ and $\eta_{2}$.
In the case where $\eta_{1} = \eta_{2} = \psi$ we can only construct a single positive three-form on $\mathcal{M}_{7}$, $\phi_3$. The holonomy of $\mathcal{M}_{7}$ will be $G_{2}$ if and only if

\begin{equation}
\nabla\phi_{3} = 0\, ,
\end{equation}

\noindent
which, from equation (\ref{eq:susytransformation}), is the case for supersymmetric backgrounds in the absence of fluxes (that is, when $\mathsf{G} = 0$). Equivalently, see appendix \ref{sec:G2appendix}, $\mathcal{M}_{7}$ will have $G_{2}$ holonomy if and only if

\begin{equation}
d\phi_{3} = 0\, , \qquad d\ast\phi_{3} = 0\, .
\end{equation}

\noindent
In the presence of fluxes, the situation is more subtle, namely there should be a connection with $G_2$ holonomy, but it is in general not the Levi-Civita one. In this case the manifold does not have $G_2$ holonomy, but still it does have a  $G_2$-structure \cite{Gauntlett:2003cy,Kaste:2003zd,House:2004pm}. The exterior derivatives $d \phi_3$, $d *\phi_3$ can be decomposed into $G_2$ representations, defining the torsion classes. In a supersymmetric compactification, the different torsion classes are related to the $G_2$ representations of the four-form flux (see \cite{Gauntlett:2003cy,Kaste:2003zd,House:2004pm} for details). 


\subsection{The underlying eight-dimensional manifold ${\mathcal M}_8$}


The manifold $\mathcal{M}_{7}$ is of course seven-dimensional. However, as it will be explained in section \ref{sec:gengeometric}, the main purpose of this work is to introduce an eight-dimensional manifold $\mathcal{M}_{8}$, whose existence is motivated by the Exceptional Generalized geometric formulation of $\mathcal{N}=1$ four-dimensional M-theory compactifications. We will be see that $\mathcal{M}_{8}$ can be related to $\mathcal{M}_{7}$ in a very natural way.   As we did in section \ref{sec:sevendimmanifolds} for seven-dimensional manifolds, here we will present some of the properties of eight-dimensional manifolds admitting nowhere vanishing spinors. 

The frame bundle of an orientable, spin, eight-dimensional manifold admits a reduction to $\mathrm{Spin}(8)$. In addition, if the manifold is equipped with an \emph{admissible} four-form $\Omega\in\Gamma\left(\mathcal{M}_{8}\right)$ (see appendix \ref{sec:Spin7appendix} for more details and further references), then the structure group is reduced to $\mathrm{Spin}(7)$. The manifold $\mathcal{M}_{8}$ has $\mathrm{Spin}(7)$ holonomy if and only if

\begin{equation}
d\Omega = 0\, .
\end{equation}

\noindent
The failure of $\Omega$ to be closed is a measure of the deviation of $\mathcal{M}_{8}$ to have $\mathrm{Spin}(7)$ holonomy. The $\mathrm{Spin}(7)$ structure can be alternatively defined by a globally defined Majorana-Weyl spinor, which implies $\mathrm{Spin}(7)$ holonomy if and only if it is covariantly constant respect to the Levi-Civita spinor connection. 

If there are two globally defined Majorana-Weyl spinors, then the structure group can be further reduced, depending on the relative properties of the spinors. If the two spinors have opposite chirality, the structure group of ${\cal M}_8$ is reduced to $G_2$, but the corresponding Riemannian metric is not irreducible. If the spinors have the same chirality and are never parallel, then the structure group is reduced to $SU(4)$. Notice that in the case of $SU(4)$ holonomy, i.e. where the spinors are covariantly constant, then the manifold os a Calabi-Yau four-fold. In the general case, where the two spinors might become parallel at some points, there is no global topological reduction further than the $Spin(7)$ one. However, as it happened in seven dimensions, there is a well-defined global reduction from the point of view of the generalized tangent space $TX\oplus T^{\ast}X$, from $\mathbb{R}^{\ast}\times Spin(8,8)$ to $Spin(7)\times Spin(7)$.


\subsection{Motivation for $\mathcal{M}_{8}$: the generalized geometric formulation}
\label{sec:gengeometric}


\noindent
A geometric formulation of the bosonic sector of eleven-dimensional Supergravity compactified down to four-dimensions was developed in \cite{Hull:2007zu,Pacheco:2008ps}, extending the idea underlying a similar formulation for the NS sector of Type-II Supergravity, based on Generalized Complex Geometry. In the latter, the diffeomorphisms and the gauge transformations of the B-field, generated respectively by vectors and one-forms, are combined into sections of the generalized tangent space, which is locally the sum of the tangent plus the cotangent space. In the former, diffeomorphisms are combined with two-form gauge transformations of the three-form gauge field $C$. In order to complete a closed orbit under the U-duality group $E_{7(7)}$\footnote{$E_{7(7)}$ is the maximally non-compact real form of the complex exceptional Lie group $E_{7}$.}, one needs however to include also the gauge transformations of the dual six-form field $\tilde C$, given by five-forms, as well as gauge transformations for the dual graviton, parameterized by the tensor product of one-forms times seven-forms \cite{Pacheco:2008ps,Coimbra:2011ky}. The total number of degrees of freedom is 56, corresponding to the fundamental representation of $E_{7(7)}$. The exceptional generalized tangent space $E$ is locally given by\footnote{To abbreviate we use $T_7$ to denote $T{\cal M}_7$.}  

\begin{equation}\label{decomposition_gl7_fund}\begin{aligned}
 E &=   (\Lambda^7 T_7)^{1/2} \otimes \left(T_7 \oplus \Lambda^2 T_7^* \oplus \Lambda^5 T_7^* \oplus (T_7^* \otimes \Lambda^7 T_7^*) \right) \ , \\
 {\bf 56} &=  {\bf 7} \oplus {\bf 21} \oplus {\bf 21} \oplus {\bf 7}\, ,
\end{aligned}\end{equation}

\noindent
where the overall volume factor gives the proper embedding in $E_{7(7)}$. This representation is also the one that combines the charges of the theory, namely momentum, M2 and M5-brane charge and Kaluza-Klein monopole charge. Note that the 
${\bf 21}$ and ${\bf 7}$ representations can be combined into the ${\bf 28}$ of $SL(8,\mathbb{R})$, corresponding to two-forms or two-vectors in eight dimensions. To be more precise, defining 

\begin{equation} \label{decomposition_gl7_V}\begin{aligned}
T_8^* &= (\Lambda^7 T_7^*)^{-1/4} \otimes (T_7^* \oplus \Lambda^7 T_7^* ) \, ,  \\
{\bf 8} &=  {\bf 7} \oplus {\bf 1}\, , 
\end{aligned}\end{equation} 

\noindent
we have
\begin{equation}\label{decomposition_sl8_fund} \begin{aligned}
 E =&\Lambda^2 T_8 \oplus \Lambda^2 T_8^* \, , \\
  {\bf 56} = & {\bf 28} \oplus {\bf 28'}  \, .
  \end{aligned}
\end{equation}

\noindent 
At each point over the seven-dimensional manifold $\mathcal{M}_{7}$, the fibre of the extended  vector bundle can be naturally decomposed in terms of the eight dimensional vector space where $Sl(8,\mathbb{R})$ acts in the fundamental representation. However, this eight dimensional vector space cannot be naturally identified, at each point $p\in\mathcal{M}_{7}$, with the tangent vector space of $\mathcal{M}_{7}$, because of the obvious dimensional mismatch. Since $Sl(8,\mathbb{R})$ is the structure group of an eight-dimensional orientable manifold, we consider that is natural to propose an eight-dimensional orientable manifold $\mathcal{M}_{8}$ such that, at each point $p\in\mathcal{M}_{8}$, carries a decomposition of the $E_{7(7)}$ appropriate representation in terms of the fundamental representation of $Sl(8,\mathbb{R})$ acting on the eight-dimensional tangent space $T_{p}\mathcal{M}_{8}$. This way, the tangent space of $\mathcal{M}_{8}$ can be \emph{connected} to the rank eight  vector bundle, with structure group $Sl(8,\mathbb{R})$, proposed in \cite{Grana:2012zn}. We will elaborate later about this connection. 

In order for $\mathcal{M}_{8}$ to carry at each point a decomposition of the appropriate $E_{7(7)}$ representation, it must be equipped, at every point, with the tensors appearing in the given decomposition. Therefore, $\mathcal{M}_{8}$ must be equipped with globally defined tensors, determined by the decomposition of the corresponding $E_{7(7)}$ representation in terms of $Sl(8,\mathbb{R})$ representations. There are three relevant $E_{7(7)}$ representations appearing in the Exceptional Generalized formulation of M-theory compactified to four dimensions, and thus have different possibilities depending on which $E_{7(7)}$ representation we consider, to wit

\begin{itemize}

\item The fundamental, symplectic, representation {\bf 56} of $E_{7(7)}$. The corresponding decomposition, in terms of $Sl(8,\mathbb{R})$ representations, is given by (\ref{decomposition_sl8_fund}), namely

\begin{equation}
E_{{\bf 56}} =  \Lambda^2 V \oplus \Lambda^2 V^{\ast}\, ,
\end{equation}

\noindent
where $V$ is an eight-dimensional real vector space. Therefore, if we wanted $\mathcal{M}_{8}$ to carry a representation of $E_{{\bf 56}}$ in terms of $Sl(8,\mathbb{R})$ representations, it should be equipped with a bivector $\beta$ field and a two-form $\omega$

\begin{equation}
\omega\in\Gamma\left(\Lambda^2 T^{\ast}\mathcal{M}_{8}\right)\, , \qquad \beta\in\Gamma\left(\Lambda^2 T\mathcal{M}_{8}\right)\, .
\end{equation}

\item The adjoint representation {\bf 133} of $E_{7(7)}$. The corresponding decomposition, in terms of $Sl(8,\mathbb{R})$ representations, is given by

\begin{equation}
E_{{\bf 133}} = \left( V \otimes V^{\ast}\right)_{0} \oplus\Lambda^4 V^{\ast}\, ,
\end{equation}

\noindent
where $V$ is an eight-dimensional real vector space and the subindex 0 denotes traceless. Therefore, if we wanted $\mathcal{M}_{8}$ to carry a representation of $E_{{\bf 133}}$ in terms of $Sl(8,\mathbb{R})$ representations, it should be equipped with the following sections

\begin{equation}
\mu\in\Gamma\left(T\mathcal{M}_{8}\otimes T^{\ast}\mathcal{M}_{8}\right)_0\, , \qquad \Omega\in\Gamma\left(\Lambda^4 T^{\ast}\mathcal{M}_{8}\right)\, .
\end{equation}

\item The {\bf 912} representation of $E_{7(7)}$. The corresponding decomposition, in terms of $Sl(8,\mathbb{R})$ representations, is given by

\begin{equation} \label{rep912}
E_{{\bf 912}} = S^2 V\oplus \left(\Lambda^{3}V\otimes V^{\ast}\right)_{0}\oplus S^2 V^{\ast}\oplus \left(\Lambda^{3}V^{\ast}\otimes V\right)_{0}\, ,
\end{equation}

\noindent
where $V$ is an eight-dimensional real vector space and $S^2$ denotes symmetric two-tensors. Therefore, if we wanted $\mathcal{M}_{8}$ to carry a representation of $E_{{\bf 912}}$ in terms of $Sl(8,\mathbb{R})$ representations, it should be equipped with the following sections

\begin{eqnarray}
g\in\Gamma\left( S^2 T\mathcal{M}_{8}\right)\, , \qquad \phi\in\Gamma\left(\Lambda^3 T\mathcal{M}_{8}\otimes T^{\ast}\mathcal{M}_{8}\right)_{0}\, , \nonumber\\
\tilde{g}\in\Gamma\left( S^2 T^{\ast}\mathcal{M}_{8}\right)\, , \qquad \tilde{\phi}\in\Gamma\left(\Lambda^3 T^{\ast}\mathcal{M}_{8}\otimes T\mathcal{M}_{8}\right)_{0}\, .
\end{eqnarray}

\end{itemize}

\noindent
In this letter we are going to consider the ${\bf 912}$ representation, since it is the relevant one to describe the moduli space of $\mathcal{N}=1$ supersymmetric M-theory compactifications. In doing so, we will be able to translate the information about the moduli space contained in the generalized $E_{7(7)}$-bundle to the tangent bundle of the eight-dimensional manifold $\mathcal{M}_{8}$, therefore giving an \emph{intrinsic} formulation in terms of tensor bundles instead of \emph{extrinsic} bundles, which generically are more difficult to handel: using the $E_{7(7)}$-bundle one can characterize the moduli space using its space of sections together with the appropriate differential conditions and equivalence relation, whereas using $\mathcal{M}_{8}$ one can characterize the moduli space through the space of sections of several of its tensor bundles, together again with the appropriate differential conditions and equivalence relation. The space of sections needed to characterize the moduli space using $\mathcal{M}_{8}$ is what we will define later to be an \emph{intermediate structure}, see definition \eqref{def:intermediatemanifold}. Therefore, we propose that the study of the moduli space of $\mathcal{N}=1$ M-theory compactifications to four-dimensional Minkowski space-time can be rephrased as the study of the moduli space of intermediate structures on the corresponding eight-dimensional manifold $\mathcal{M}_{8}$. The appearance of the {\bf 912} representation can be justified as follows. 

The presence of the supersymmetry spinors $\eta_{1}, \eta_{2}\in\Gamma\left( S^{\mathbb{R}}_{7}\right)$ in $\mathcal{M}_{7}$ implies a global reduction of an appropriate vector bundle which is an extension of the tangent space. In order to perform the reduction, we have to identify the complex $Spin(7)$-spinor $\eta = \eta_{1} + i\eta_{2}$ as a complex Weyl $Spin(8)$ spinor. As explained in the previous section, with respect to $Spin(7)$, the real and imaginary parts of the spinor define each a $G_2$ structure and together they define a reduction of the structure group of the generalized bundle $\mathbb{E} = T\mathcal{M}_{7}\oplus T^{\ast}\mathcal{M}_{7}$ to $G_{2}\times G_{2}$. In $Sl(8,\mathbb{R})$, the real and imaginary parts of the complex $Spin(8)$ spinor $\eta$ define a pair of $Spin(7)$ structures. In $E_{7(7)}$ the complex spinor transforms in the fundamental of $SU(8)$, and defines a single $SU(7)$ structure, since within $E_{7(7)}$ the complex spinor transforms in the ${\bf 8}$ of $SU(8)$, and is stabilized by an $SU(7)$ subgroup.  This $SU(7)$ structure can be equivalently defined by a nowhere vanishing section of the ${\bf 912}$ representation of $E_{7(7)}$, which decomposes into the $Sl(8,\mathbb{R})$ representations as in (\ref{rep912}).  The following object, constructed from the internal spinor $\eta$, is indeed stabilized by $SU(7)\subset SU(8) \subset E_{7(7)}$ \cite{Pacheco:2008ps}

\begin{equation}\label{embedding_spinor_912}
 \psi =  (2 \eta \otimes \eta, 0, 0,0) \, .  
\end{equation}

\noindent
Using an $Sl(8,\mathbb{R})$ metric $g_8$, this object has the following $Sl(8,\mathbb{R})$ representations \cite{Grana:2012zn}

\begin{equation} \label{phi912}
\psi= ( {\rm Re} c \, g_8^{-1}, g_8 \cdot ({\rm vol}^{-1}_8 \llcorner  {\rm Re} \phi_4) , {\rm Im} c \, g_8 , g_8^{-1} \cdot  {\rm Im} \phi_4)\, ,
\end{equation}

\noindent
where $c= \eta^T \eta$ and the four-form $\phi_4$ is

\begin{equation}
\label{eq:phi4}
\phi_{4} = \eta^{T} \gamma_{(4)} \eta \, .
\end{equation}

\noindent
In terms of the 7+1 split in (\ref{decomposition_gl7_V}), this is

\begin{equation}
\label{eq:Omegabundledecomposition}
\phi_{4} =\rho_8 \wedge \phi_3  + \ast_7 \phi_3 \, ,
\end{equation}

\noindent
where $\rho_8$ is a one-form along the ${\bf 1}$ in (\ref{decomposition_gl7_V}), and $\phi_3$ is a complex three-form which reduces to a real three-form in the $G_2$-structure case (i.e.\ when $\eta$ is Majorana), given by (\ref{eq:3forma}).

\noindent
In sections \ref{sec:8manifold} and on we will analyze the geometric properties of $\mathcal{M}_{8}$ and connect the geometric structures defined on it to the supersymmetry spinors $\eta_{1}$ and $\eta_{2}$, globally defined on $\mathcal{M}_{7}$ when the compactification to four-dimensions preserves some supersymmetry off-shell. Among other things, the goal of this letter is to study the possible role of $\mathcal{M}_{8}$ in relation to the internal manifolds appearing in $\mathcal{N}=1$ supersymmetric M-theory compactifications and, since $\mathcal{M}_{8}$ is eight dimensional, also study if $\mathcal{M}_{8}$ is an admissible internal space for F-theory compactifications. For the latter, we give in section \ref{sec:Ftheoryreview} a very brief review of the type of eight-dimensional manifolds that appear as internal spaces in F-theory. 


\subsection{Connection to F-theory}
\label{sec:Ftheoryreview}


Being $\mathcal{M}_{8}$ eight-dimensional, the natural question is that if it is an admissible internal space for F-theory \cite{Vafa:1996xn} compactifications down to four dimensions (see \cite{Morrison:1996na,Morrison:1996pp,Denef:2008wq} for more details and further references). F-theory compactifications to four-dimensions are defined through M-theory compactifications to three dimensions on an eight-dimensional manifold $X$. In order to have $\mathcal{N}=1$ supersymmetry in four dimensions, the compactification theory in three-dimensions must have $\mathcal{N}=2$ supersymmetry. This imposes a constraint on the eight-dimensional internal space $X$, which must be a Calabi-Yau four-fold\footnote{We refer to flux-less compactifications. If we include a non-vanishing $\mathsf{G}_{4}$ flux, then $X$ is a \emph{conformal} Calabi-Yau four-fold, which does not have $SU(4)$-holonomy anymore.}, that is, a $SU(4)$-holonomy eight-dimensional manifold. In order to be able to \emph{appropriately lift} the three-dimensional effective $\mathcal{N}=2$ Supergravity theory to four dimensions, the internal space must be in addition elliptically fibered, that is, it must be of the form

\begin{equation}
X\xrightarrow{\pi} \mathcal{B}\, ,
\end{equation} 

\noindent
where $\mathcal{B}$ is the base space, which should be a three-complex-dimensional K\"ahler manifold, and the fibre $\pi^{-1}\left(b\right)$ at every $b\in\mathcal{B}$ is an elliptic curve, possibly singular. Therefore, if we want the proposed $\mathcal{M}_{8}$ to be an admissible internal manifold for supersymmetric compactifications of F-theory, it must be Calabi-Yau and elliptically fibered. As mentioned before, for the structure group of an eight-dimensional manifold $X$ to be reduced to $SU(4)$, it must be equipped with two Majorana-Weyl spinors of the same chirality and linearly independent at every point $p\in X$, which are covariantly constant if and only if $X$ has $SU(4)$ holonomy. If $X$ is equipped with two Majorana-Weyl spinors of the same chirality but which are parallel at some points then there is no global reduction of the structure group on the tangent space to $SU(4)$, but there is a global reduction in $TX\oplus T^{\ast}X$ from $\mathbb{R}^{\ast}\times Spin(8,8)$ to $Spin(7)\times Spin(7)$. To the best of our knowledge, F-theory compactified in manifolds with $Spin(7)\times Spin(7)$ structure structure has not been fully analyzed yet. Progress in this direction can be found in reference \cite{Tsimpis:2005kj}.

If we drop the requirement of $\mathcal{N}=2$ supersymmetry in three dimensions, and demand only minimal supersymmetry in three dimensions, then $X$ is not forced to be a Calabi-Yau four-fold but a Spin(7)-holonomy manifold, and we should not expect in principle a supersymmetric theory in for dimensions. However, and remarkably enough,  \cite{Bonetti:2013fma,Bonetti:2013nka} claim that F-theory compactified on certain $\mathrm{Spin}(7)$-holonomy orbifold (constructed by orbifolding a Calabi-Yau four-fold), where one dimension has the topology of an interval, give in the limit of infinite length of this interval, a supersymmetric $\mathcal{N}=1$ theory in four dimensions. Therefore we will consider in this letter that an elliptically fibered, eight-dimensional, $Spin(7)$ manifold is an admissible internal space for F-theory compactifications. In particular, we will find that the eight-dimensional manifold $\mathcal{M}_{8}$ can be elliptically fibered and it is equipped, under some mild assumptions, with a $Spin(7)$ structure on the frame bundle, or more generally, a $Spin(7)\times Spin(7)$ structure on the generalized bundle. Hence, $\mathcal{M}_{8}$ arises as  a plausible compactification space for F-theory, which in addition can be related in a precise way to $G_{2}$-structure seven-dimensional manifolds.  We will leave the analysis of the holonomy of $\mathcal{M}_{8}$ and its preferred submanifolds to a forthcoming project \cite{holonomyMCM}.


\section{Linear algebra of positive and admissible forms}
\label{sec:G2in7}


The first part of this section is devoted to introducing some linear algebra results regarding the definition of  topological $G_{2}$-structures that will be useful later on. 
More details can be found in appendix \ref{sec:G2appendix}. The second 
part of this section studies the relation, at the linear algebra level, between the differential forms associated to topological $G_{2}$-structures and topological $Spin(7)$-structures.

\begin{definition}
\label{def:nondegsymp}
Let $V$ be an $n$-dimensional vector space and let $\omega\in\Lambda^{q}V^{\ast}$ be a $q$-form.  $\omega$ is said to be \emph{non-degenerate} if the following holds: 

\begin{equation}
\forall\,\, v \in V,\,\, i_{v} \omega =0 \Rightarrow v =0\, .
\end{equation}

\end{definition}

\noindent
In other words, a non-degenerate $q$-form provides an injective map from $V$ to the vector space of $(q-1)$-forms $\Lambda^{q-1}V^{\ast}$. Let $\omega\in\Omega^{q}\left(\mathcal{M}\right)$ be a $q$-form defined on a differentiable manifold $\mathcal{M}$. Then $\omega$ is said to be non-degenerate if it is non-degenerate at every point $p\in\mathcal{M}$. The previous notion of non-degeneracy is extensively used in the context of multisymplectic geometry \cite{JAZ:4974756}. However, in references \cite{2000math.....10054H,2001math......7101H}, a different notion of non-degeneracy, called   \emph{stability}, was introduced by Hitchin. The definition goes as follows.

\begin{definition}
\label{def:stabelform}
Let $V$ be an $n$-dimensional vector space and let $\omega\in\Lambda^{q}V^{\ast}$ be a $q$-form.   $\omega$ is said to be \emph{stable} if it lies in an open orbit of the action of the group $GL\left(V\right)$ on $\Lambda^{q}V^{\ast}$.
\end{definition}

\noindent
Let $\omega\in\Omega^{q}\left(\mathcal{M}\right)$ be a $q$-form defined on a differentiable manifold $\mathcal{M}$. Then $\omega$ is said to be stable if it is stable at every point $p\in\mathcal{M}$. Although the notion of stability can be defined for any form, it can be shown that stable forms only occur in certain dimensions and for certain $q$-forms. In particular, aside from cases $q=1,2$, stability can only occur for three-forms (and their Hodge-duals), in six, seven and eight dimensions. A manifold equipped with a stable form $\omega$ has its structure group reduced to the stabilizer group of $\omega$. It is clear that the notion of stability  in general is not equivalent to the notion of non-degeneracy, as defined in definition \ref{def:nondegsymp}. For instance, in even dimensions, the notion of non-degenerate two-form is equivalent to the notion of  stable two-form, since the general linear group has only one open orbit when acting on $\Lambda^{2} V^{\ast}$, and this orbit consists exactly of the non-degenerate two forms. However, in odd dimensions a two-form can never be non-degenerate yet it can be stable. 
 
\subsection{Positive forms on seven-dimensional vector spaces} 
 
Since we are interested in seven-dimensional manifolds with $G_{2}$ structure, we will focus now on the case of three-forms in seven dimensions. The reason is explained in appendix \ref{sec:G2appendix}: a seven-dimensional manifold has the structure group of its frame bundle reduced from $GL\left(7,\mathbb{R}\right)$ to the compact real form $G_{2}$ of the complex exceptional Lie group $G_{2}^{\mathbb{C}}$ if and only if it can be equipped with a globally defined,   \emph{positive} three-form. A positive form is a particular case of stable three-form. The three-form $\phi_0$ on $\mathbb{R}^{7}$ we define now is positive.

\begin{definition}\label{def:phi0}
Let $\left(x_{1},\hdots,x_{7}\right)$ be coordinates on $\mathbb{R}^{7}$. We define a three-form $\phi_{0}$ on $\mathbb{R}^{7}$ by 
\begin{equation}
\label{eq:phi0}
\phi_{0} = dx_{123} + dx_{145} + dx_{167} + dx_{246} - dx_{257} - dx_{347} - dx_{356}  \, ,
\end{equation}

\noindent
where $dx_{ij\hdots l}$ stands for  $dx_{i}\wedge dx_{j}\wedge \cdots\wedge dx_{l}$. The subgroup of $GL\left(7,\mathbb{R}\right)$ that preserves $\phi_{0}$ is the compact real form $G_{2}$ of the exceptional complex Lie group $G_{2}^{\mathbb{C}}$, which also fixes the euclidean metric $g_{0} = dx_{1}^{2} + \cdots + dx_{7}^2$, the orientation on $\mathbb{R}^{7}$ (that is, $G_2\subset SO(7)$). Further $G_2$ fixes the four-form $\tilde \phi_{0}$,
\begin{equation}
\tilde{\phi}_{0} = dx_{4567} + dx_{2367} + dx_{2345} + dx_{1357} - dx_{1346} - dx_{1256} - dx_{1247}  \, .
\end{equation}

\noindent
Notice that $\tilde{\phi}_{0} = \ast \phi_{0}$ where $\ast$ is the Hodge-dual operator associated to $g_{0}$. 

\end{definition} 

\noindent
Notice that every three-form $\phi\in \Lambda^{3} V^{\ast}$ defines a symmetric bilinear form $\mathfrak{B}: V\times V\to \mathbb{R}$ as follows:
\begin{eqnarray}
\label{eq:Bilinear}
\mathcal{V}\cdot \mathfrak{B}\left( v, w\right) = \frac{1}{3!}\,\iota_{v}\phi\wedge\iota_{w}\phi\wedge\phi\, ,\qquad v, w \in V\, ,
\end{eqnarray}

\noindent
where $\mathcal{V}$ is the seven-dimensional volume form in $V$. If $\phi$ is a no-where vanishing non-degenerate three-form, then (\ref{eq:Bilinear}) is a non-degenerate bilinear form, in the sense that

\begin{equation}
\mathfrak{B}\left( v, v\right)\neq 0\, , \qquad \forall\,\, v\in V-\{0\}\, .
\end{equation}

\noindent
In addition, if $\phi$ is positive then $\mathfrak{B}: V\times V\to \mathbb{R}$ is a  positive definite symmetric bilinear form on $V$. If we take $\phi = \phi_{0}$ we obtain

\begin{equation}\label{eq:innpr}
\mathcal{V}\cdot \mathfrak{B}_{0}\left( v,w\right) = \frac{1}{3!}\,\iota_{v}\phi_{0}\wedge\iota_{w}\phi_{0}\wedge\phi_{0} = g_{0}\left( v,w\right)\mathcal{V}\, ,\qquad \forall\,\, v, w \in V\, ,
\end{equation}

\noindent
and therefore $g_{0}$ is the metric (i.e., inner product) naturally induced on $V$ by the three form $\phi_{0}$.
 
\begin{definition}
Let $\left(x_{1},\hdots,x_{7}\right)$ be coordinates on $\mathbb{R}^{7}$. We define a three-form $\phi_{1}$ on $\mathbb{R}^{7}$ by 
\begin{equation}
\label{eq:phi1}
\phi_{1} =  -dx_{127} + dx_{145} - dx_{135} + dx_{146} + dx_{236} - dx_{245} - dx_{347} + dx_{567} \, ,
\end{equation}

\noindent
where $dx_{ij\hdots l}$ stands for $dx_{i}\wedge dx_{j}\wedge \cdots\wedge dx_{l}$. The subgroup of $GL\left(7,\mathbb{R}\right)$ that preserves $\phi_{1}$ is the non-compact real form $G^{\ast}_{2}\subset SO(4,3)$ of the exceptional complex Lie group $G^{\mathbb{C}}_{2}$, which also fixes the indefinite metric $g_{1} = dx_{1}^{2} + dx_{2}^{2} + dx_{3}^{2}+dx_{4}^{2}-dx_{5}^{2}-dx_{6}^{2} -dx_{7}^2$, the orientation on $\mathbb{R}^{7}$ and the four-form $\ast \phi_{1}$,
\begin{equation}
\tilde{\phi}_{1} = -dx_{3456} + dx_{2367} - dx_{2467} + dx_{2357} + dx_{1457} - dx_{1367} - dx_{1256} + dx_{1234}\, .
\end{equation}

\noindent
Notice that $\tilde{\phi}_{1} = \ast \phi_{1}$ where $\ast$ is the Hodge-dual operator associated to $g_{1}$. 

\end{definition}   
  
\noindent
In this case we have 

\begin{equation}
\mathcal{V}\cdot\mathfrak{B}_{1}\left( v,w\right) = \frac{1}{3!}\iota_{v}\phi_{1}\wedge\iota_{w}\phi_{1}\wedge\phi_{1} = g_{1}\left( v,w\right)\mathcal{V}\, ,\qquad \forall\,\, v, w \in V\, ,
\end{equation}
and therefore $g_{1}$ is the metric naturally induced in $V$ by the three form $\phi_{1}$.

\bigskip
\noindent
We now define positive three-forms and display some of their properties.
\begin{definition}\label{def:positive}
Let $V$ be an oriented seven-dimensional vector space. A three-form $\phi\in\Lambda^{3} V^{\ast}$ is said to be \emph{positive} if there exists an 
 oriented\footnote{This means that $f$ preserves the orientations.} isomorphism $f:\mathbb{R}^{7}\to V$ such that $\phi_{0} = f^{\ast} \phi$. We denote by $\mathcal{P}_{V}$ the set of positive three-forms in $\phi\in\Lambda^{3} V^{\ast}$ .
\end{definition}

\noindent
A positive form induces an inner product on $V$, as a 
consequence of $G_2\subset SO(7)$.

\begin{prop}\label{prop:35}
Let $V$ be a seven-dimensional vector space. Then the set $\mathcal{P}_{V}$ of positive forms is an open subset of $\Lambda^{3} V^{\ast}$.
\end{prop}
 
\begin{proof}Since $\mathcal{P}_{V}$ is defined as those forms in $\Lambda^{3} V^{\ast}$ such that there exists an oriented isomorphism relating them to $\phi_{0}$, and $\phi_{0}$ is stabilized by $G_{2}$, we conclude that

\begin{equation}
\mathcal{P}_{V} \simeq \frac{GL_{+}\left( 7,\mathbb{R}\right)}{G_{2}}\, ,
\end{equation}

\noindent
and thus, since $G_{2}$ is a Lie subgroup of $GL_{+}\left( 7,\mathbb{R}\right)$, $\mathcal{P}_{V}$ is an homogeneous manifold of dimension $\mathrm{dim}\, \mathcal{P}_{V} = \mathrm{dim}\, GL_{+}\left( 7,\mathbb{R}\right) - \mathrm{dim}\, G_{2} = 49 - 14 = 35$. $\Lambda^{3} V^{\ast}$ is a vector space  of dimension $\mathrm{dim}\, \Lambda^{3} V^{\ast} = \frac{7!}{4!3!} = 35$. Hence, $\mathcal{P}_{V}$ is a submanifold of $\Lambda^{3} V^{\ast}$ of the same dimension as $\Lambda^{3} V^{\ast}$, and therefore it must be an open set in $\Lambda^{3} V^{\ast}$. \end{proof}

\noindent
By the proposition above,  a positive form is stable, since it belongs to an open orbit. However, the converse is not true, due to the fact there exists more than one open orbit. The description of the set of stable forms on $\mathbb{R}^{7}$ is the following  \cite{LieEngel,Engel,Reichel}.

\begin{itemize}

\item The general group $GL\left( 7,\mathbb{R}\right)$ acting on $\Lambda^{3} (\mathbb{R}^{7})^{\ast}$ has exactly two open orbits $\Lambda_{0}$ and $\Lambda_{1}$, each of which is disconnected  and can be characterized as follows:

\begin{itemize}

\item $\Lambda_{0}$ contains $\phi_{0}$, therefore any other form $\phi\in\Lambda_{0}$ is stabilized by a group conjugate to the real compact form $G_{2}$ of $G^{\mathbb{C}}_{2}$.

\item $\Lambda_{1}$ contains $\phi_{1}$, therefore any other form $\phi\in\Lambda_{0}$ is stabilized by a group conjugate to the non-compact real form $G^{\ast}_{2}$ of $G^{\mathbb{C}}_{2}$.

\end{itemize}

\item Each open orbit consists of two connected components, namely $\Lambda^{\pm}_{0}$ and $\Lambda^{\pm}_{1}$, which are given by the action of $GL_{\pm}\left( 7,\mathbb{R}\right)$
on $\phi_{0}$ and $\phi_{1}$
respectively . Here $GL_{+}\left( 7,\mathbb{R}\right)$ and  $GL_{-}\left( 7,\mathbb{R}\right)$ are the elements of $GL\left( 7,\mathbb{R}\right)$ of positive and negative determinant respectively.

\item The set of stable three-form $\phi\in\Lambda^3 V^{\ast}$  is given by the disjoint union $\Lambda^{+}_{0}\cup \Lambda^{-}_{0} \cup \Lambda^{+}_{1}\cup \Lambda^{-}_{1}$. The positive forms are exactly those contained in $\Lambda^{+}_{0}$.
\end{itemize}

\begin{prop}
Let $V$ be a seven-dimensional vector space and let $\mathcal{N}_{V}$ and $\mathcal{S}_{V}$ respectively denote the set of nondegenerate three-forms and the set of stable three-forms. Then, $\mathcal{S}_{V}\subset \mathcal{N}_{V}$, that is, $\mathcal{S}_{V}$ is a subset $\mathcal{N}_{V}$. In addition, $\mathcal{S}_{V}\neq \mathcal{N}_{V}$, that is, there exist non-degenerate three-forms that are not stable.
\end{prop}

\begin{proof} Since every stable form is related to $\phi_{0}$ or $\phi_{1}$ by an isomorphism, to show $\mathcal{S}_{V}\subset \mathcal{N}_{V}$ it is enough to check that $\phi_{0}$ and $\phi_{1}$ are non-degenerate. This can be checked by explicit calculation. In order to see now that there are more degenerate three-forms than stable three-forms, we are going to proceed by a parameter count\footnote{We thank Dominic Joyce for a private communication regarding this issue.}.

\begin{itemize}

\item Dimension of the space of three-forms: $\dim\, \Lambda^3 V^{\ast} = 35$.

\item Dimension of the space of three forms $\phi\in  \Lambda^3 V^{\ast}$ such that $\iota_{v}\phi = 0$, where $v\in V$ is a fixed unit norm vector = Dimension of three-forms on $\mathbb{R}^{6}$ = 20.

\item Dimension of the space of unit vectors in $V$ = 6.

\item Dimension of the space of three forms $\phi\in  \Lambda^3 V^{\ast}$ such that $\iota_{v}\phi = 0$, for some unit vector $v\in V$ = Dimension of three-forms on $\mathbb{R}^{6}$ + Dimension of unit vectors in $V$ = 26.

\end{itemize}

\noindent
Therefore $\Lambda^3 V^{\ast} -\mathcal{N}_{V}$,
the space of degenerate three-forms, has dimension 26. On the other hand, as seen above 
$\mathcal{S}_{V}$ is not connected, hence 
the space $\Lambda^3 V^{\ast} -\mathcal{S }_{V}$ of three-forms which are not stable must have a component of codimension one, i.e. dimension 34. 
(If $\Lambda^3 V^{\ast} -\mathcal{S }_{V}$ had codimension $\ge 2$, then $\mathcal{S }_{V}$ would be connected.)
Since $26 < 34$, we conclude that there must be 
a three-form that is not stable but which is
non-degenerate.
\end{proof}

\subsection{Relation between admissible and positive forms}
\label{sec:admpos}

Positive three-forms (definition \ref{def:positive}) on seven-dimensional vector spaces are closely related to  admissible four-forms (definition \ref{def:admissible}) on eight-dimensional vector spaces.
We first show how to pass from  positive three-forms   to admissible four-forms.

\begin{lemma}
\label{lem:phitoOmega}
Let $V$ be an eight-dimensional oriented vector space and $H$ a seven-dimensional oriented subspace. Let $\phi$ be a positive three-form on $H$, and $v\in V$ a vector transverse to $H$ and inducing\footnote{The
orientation on $H$ induced by the one on $V$ and the vector $v$ is defined as follows:
a  basis  $w_1,\dots,w_7$     of $H$ is declared to be compatible with the induced orientation if{f} $v,w_1,\dots,w_7$ a  basis of $V$ compatible with its orientation.} the given
 orientation on $H$. Then\footnote{Here we slightly abuse notation, denoting  by the same symbol $\phi \in \Lambda^3H^*$ and its extension to
a three-form on $V$ annihilating the vector $v$.}
\begin{equation*}
\Omega = v^{\flat}\wedge  \phi+ \ast\left(v^{\flat}\wedge  \phi\right)\, ,
\end{equation*}

\noindent
is an admissible four-form on $V$, where $v^{\flat}$ and the Hodge star $\ast$ are taken w.r.t. the unique inner product $g$ on $V$ which on $H$ agrees  with the one associated to $\phi$, and which satisfies $||v||=1$, $v\perp H$. 
\end{lemma}

\begin{proof}
Denote by $e_1,\dots,e_8$ the standard basis of $\mathbb{R}^{8}$. On the 7-dimensional subspace $Span\{e_2,\dots,e_8\}$,
consider the three-form $\phi_0$ given as in Def. \ref{def:phi0} (but shifting all the indices by one, so that they lie in the range $2,\dots,8$). Then $e_1^{\flat}\wedge  \phi\in \wedge^4(\mathbb{R}^{8})^*$ equals the first eight terms of the four-form $\Omega_0$ appearing in equation \eqref{eq:fourform}, while $\ast_7\phi_0=\ast\left(e_1^{\flat}\wedge  \phi_{0}\right)$ equals the remaining eight terms of  $\Omega_0$, where $\ast_7$ denotes the Hodge-star on $Span\{e_2,\dots,e_8\}$. Hence, $e_1^{\flat}\wedge  \phi_0+\ast\left(e_1^{\flat}\wedge  \phi_{0}\right)=\Omega_0$.

There is an orientated isometry $\tau \colon H \to \mathbb{R}^{7}$ identifying $\phi$ with $\phi_0$.
Denote by $f_2,\dots,f_8$ the orthonormal basis of $H$ corresponding to the standard basis of $\mathbb{R}^{7}$ under $\tau$, and define  $f_1:=v$. Then $f_1,\dots,f_8$ is a basis of $V$ which is orthonormal w.r.t. the inner product $g$ and compatible with the orientation. The coordinate map $V\to \mathbb{R}^{8}$ is an oriented isometry which restricts to $\tau$, and the above argument\footnote{Notice that the Hodge star depends on the inner product and the orientation.} on $\mathbb{R}^{8}$ shows that $\Omega$ is an admissible form.
\end{proof}

\begin{remark}
\label{rem:innonV} In Lemma \ref{lem:phitoOmega},
the inner product on $V$ associated to $\Omega$ is exactly  $g$. This follows from the fact that the inner product on $\mathbb{R}^{8}$ associated to $\Omega_0$ is the standard inner product. Further, in Lemma \ref{lem:phitoOmega} the term $ \ast\left(v^{\flat}\wedge  \phi\right)$ equals  $ \ast_7 \phi $, the Hodge-dual of $\phi$ w.r.t. the metric on $H$ induced by $\phi$, extended to a form on $V$ with kernel $ \mathbb{R}v$.
\end{remark}

\noindent
Conversely, we now show how to pass from admissible four-forms to positive three-forms.

\begin{lemma}\label{lem:anyv}
 Let $V$ be an eight-dimensional oriented vector space and $\Omega$ an admissible four-form.
 For all non-zero 
 $v\in  V $, and all $7$-dimensional subspaces $H\subset V$ transverse to $v$, 
 the form
$(\iota_v\Omega)|_{H}$
 is a positive three-form on $H$ (where $H$ has the orientation induced by the orientation on $V$ and by $v$).
\end{lemma}
 \begin{proof} 
We may assume that $V=\mathbb{R}^{8}$ and that $
\Omega=\Omega_0$, the four-form given in eq. \eqref{eq:fourform}. 
Denote by $e_1,\dots,e_8$ the standard basis of $\mathbb{R}^{8}$.
Notice that $(\iota_{e_1}\Omega_0)|_{Span\{e_2,\dots,e_8\}}$ coincides with the three-form $\phi_0$ given in Def. \ref{def:phi0} (upon a shift of indices). 
For any $7$-dimensional subspace $H'$ transverse to $e_1$, the isomorphisms $H'\to Span\{e_2,\dots,e_8\}$ obtained restricting the orthogonal projection
identifies $(\iota_{e_1}\Omega_0)|_{H'}$ and 
 $(\iota_{e_1}\Omega_0)|_{Span\{e_2,\dots,e_8\}}$, therefore the former is   a positive three-form.

Consider now the action of $Spin(7)$ on $\mathbb{R}^{8}$, obtained by restriction to $Spin(7)\subset SO(8)$ of the usual action of $SO(8)$ on $\mathbb{R}^{8}$. Recall that this action of $Spin(7)$   preserves $\Omega_0$.
Further, this action of $Spin(7)$ is transitive on the unit sphere $S^7$,
as explained in \cite[\S 3.1, Remark 3]{Bryant}.
Hence, for every $v\in S^7$, there is $A\in Spin(7)$ with $Ae_1=v$.
The form $$A^*(\iota_v\Omega_0)=\iota_{A^{-1}v}(A^*\Omega_0)=\iota_{e_1}\Omega_0$$ is positive once restricted to $A^{-1}(H)$, by the beginning of this proof, hence
$\iota_v\Omega_0$ is positive once restricted to $H$.

 To show that the same holds for all non-zero vectors, notice that if $r>0$, then $\iota_{v}\Omega_0$ and $\iota_{rv}\Omega_0$ are $GL(8)$-related (through the dilation by $\sqrt[3]{r}$) therefore the restriction to $H$ of the latter form is also positive.
\end{proof}

\begin{remark}\label{rem:metricH}
In the setting of Lemma \ref{lem:anyv}, denote by $g$ the metric on $V$ induced by the admissible form $\Omega$, and
assume that $||v||=1$ and $H=v^{\perp}$. Then the metric on $H$ induced by the positive form $(\iota_v\Omega)|_{H}$ is the restriction of $g$.
 Indeed, since the action of $Spin(7)$ on $\mathbb{R}^{8}$ 
 preserves $\Omega_0$ as well as the metric $g$, and  is transitive
 on $S^7$,
 it is enough to check this statement for $v=e_1$, for which  it is clearly true.
\end{remark}
 
\noindent
  There is a bijective correspondence between admissible and positive forms which are compatible with a given, fixed metric. 

\begin{prop}
\label{prop:bijection}
Let $V$ be an oriented eight-dimensional vector space with a fixed   inner product $g$. Fix $v\in V$ with $||v||=1$. Denote by $i^*g$ the restricted inner product on the seven-dimensional subspace $H:=v^{\perp}$ (endowed with the orientation induced by the one  on $V$ and by $v$).
 There is a bijection
\begin{align*}
A\colon \{\text{admissible four-forms on $V$ inducing $g$}\}
&\to \{\text{positive three-forms on $H$ inducing $i^*g$}\}\\
\Omega &\mapsto (\iota_v\Omega)|_H
\end{align*}
whose inverse is given by
\begin{align*}
B\colon \{\text{positive three-forms on $H$ inducing $i^*g$}\}
&\to \{\text{admissible four-forms on $V$ inducing $g$}\}\\
\phi &\mapsto v^{\flat}\wedge  \phi+ \ast\left(v^{\flat}\wedge  \phi\right)
\end{align*}
\end{prop}

\begin{proof}
Notice first that the map $A$ is well-defined by lemma \ref{lem:anyv} and remark \ref{rem:metricH}. Similarly, the map $B$ is well-defined by lemma \ref{lem:phitoOmega} and remark \ref{rem:innonV}.

Clearly $A\circ B=Id$. Instead of showing directly that $B\circ A=Id$, we proceed as follows.

Recall from appendix \ref{sec:Spin7appendix} that the admissible four-forms on $\mathbb{R}^{8}$ are the elements of the orbit of
$\Omega_{0}$ (see eq. \eqref{eq:fourform}) under the natural action of 
$GL_{+}\left(8,\mathbb{R}\right)$ on $\Lambda^{4} (\mathbb{R}^{8})^{\ast}$, and that the stabilizer of $\Omega_{0}$ is the subgroup $Spin(7)$. Hence the admissible four-forms on $\mathbb{R}^{8}$ whose associated inner product is the standard one are given by the $SO(8)$-orbit through $\Omega_{0}$. Therefore we obtain diffeomorphisms
$$\{\text{admissible four-forms on $V$ inducing $g$}\}\cong SO(8)/Spin(7)\cong \mathbb{R}P^7.$$

Similarly, by \S \ref{sec:G2in7} and
appendix \ref{sec:G2appendix},
  the positive three-forms on $\mathbb{R}^{7}$ are the elements of the orbit of
$\phi_{0}$ (see eq. \eqref{def:phi0}) under the natural action of 
$GL_{+}\left(7,\mathbb{R}\right)$, and   the stabilizer of $\phi_{0}$ is $G_2$. Hence positive three-forms on $\mathbb{R}^{7}$ whose associated inner product is the standard one are given by the $SO(7)$-orbit through $\phi_{0}$. Therefore  we obtain diffeomorphisms

\begin{equation}
\{\text{positive three-forms on $H$ inducing $i^*g$}\}\cong SO(7)/G_2\cong \mathbb{R}P^7\, .
\end{equation}

\noindent
We conclude that  both the domain and codomain of $A$ are the 7-dimensional real projective space $\mathbb{R}P^7$.

The equation $A\circ B=Id$ implies that $B$ is an injective immersion, therefore by dimension count a local diffeomorphism, and hence   a diffeomorphism onto its image. The image of $B$ 
is open (since $B$ is a local diffeomorphism) and closed (since  the domain of $B$ is compact), hence it must be the whole of the codomain of $B$. This shows that $B$ is surjective as well, hence bijective, with inverse $A$.
\end{proof}

\begin{cor}
\label{cor:admissible3}
Let $V$ be an eight-dimensional oriented vector space with inner product $g$. Let $\phi \in \wedge^3 V^*$ and a unit vector $v\in V$ such that $ker(\phi)=\mathbb{R}v$. Define $H:=v^{\perp}$.

\begin{enumerate}

\item If $ v^{\flat}\wedge  \phi+ \ast\left(v^{\flat}\wedge  \phi\right)$ is an admissible four-form on $V$ (not necessarily inducing the inner product $g$), then 
$\phi|_H$ is a positive three-form on $H$.

\item $ v^{\flat}\wedge  \phi+ \ast\left(v^{\flat}\wedge  \phi\right)$ is an admissible four-form on $V$ inducing $g$ if and only if $\phi|_H$ is a positive three-form inducing $i^*g$.
\end{enumerate}
\end{cor}

\begin{proof}
1. follows from lemma \ref{lem:anyv}, and 2) from proposition \ref{prop:bijection}. Both use that the contraction with $v$ of $v^{\flat}\wedge  \phi+ \ast\left(v^{\flat}\wedge  \phi\right)$ is $\phi$.
\end{proof}


\section{Decrypting the 8-dimensional manifold}
\label{sec:8manifold}


As explained in section \ref{sec:EGGM}, in  exceptional generalized geometry a compactification with four-dimensional off-shell ${\cal N}=1$ supersymmetry is described by a reduction of the structure group $E_{7(7)}$  of a vector bundle in the {\bf 912} representation. The decomposition of the {\bf 912} representation of $E_{7(7)}$ in terms of the fundamental representation of $Sl\left(8,\mathbb{R}\right)$ (the group acting on the tangent bundle of an eight-dimensional oriented  manifold ${\cal M}_8$)  is given in (\ref{rep912}). Our starting point will therefore be 
a differentiable manifold $\mathcal{M}_{8}$ equipped with the following four different global sections, not necessarily everywhere non-vanishing, 
\begin{itemize}

\item A global section $g\in\Gamma\left(S^2 T^{\ast}\mathcal{M}_{8}\right)$.

\item A global section $\tilde{g}\in\Gamma\left(S^2 T\mathcal{M}_{8}\right)$.

\item A global section $\phi\in\Gamma\left(\Lambda^3 T^{\ast}\mathcal{M}_{8}\otimes T\mathcal{M}_{8}\right)_0$.

\item A global section $\tilde{\phi}\in\Gamma\left(\Lambda^3 T\mathcal{M}_{8}\otimes T^{\ast}\mathcal{M}_{8}\right)_0$.

\end{itemize}

\noindent
We will require $\mathcal{M}_{8}$ to be equipped with a volume form, since we want to be able to define integrals and also since we want the structure group to be reduced to $Sl(8,\mathbb{R})$, in order for $\mathcal{M}_{8}$ to carry the corresponding $E_{7(7)}$ representation at every point $p\in\mathcal{M}_{8}$. Hence $\mathcal{M}_{8}$ is oriented.  

In addition, we want to consider compactifications with off-shell supersymmetry, and therefore we have to require $\mathcal{M}_{8}$ to be equipped with a Riemannian metric $\mathsf{g}$ in such a way that $\left(\mathcal{M},\mathsf{g}\right)$ is a spin manifold. Notice that $\mathcal{M}_{8}$ is equipped with a global section $g\in\Gamma\left(S^2 T^{\ast}\mathcal{M}_{8}\right)$. Therefore it is natural to take $g$ to be a Riemannian metric, and $\tilde{g}$ to be the inverse of $g$, which is indeed the case if the ${\bf 912}$ structure is defined from a $Spin(8)$ spinor $\eta$, see equation (\ref{phi912}).  We introduce thus the following definition.

\begin{definition}
\label{def:intermediatemanifold}
Let $\mathcal{M}_{8}$ be an eight-dimensional, oriented, differentiable manifold. We say that $\mathcal{M}_{8}$ is an \emph{intermediate manifold }if it is equipped with the following data
\begin{itemize}

\item A Riemannian metric $g\in\Gamma\left( S^{2} T^{\ast}\mathcal{M}_{8}\right)$ such that $\left(\mathcal{M}_{8},g\right)$ is a spin manifold.

\item A global section $\phi\in\Gamma\left(\Lambda^3 T^{\ast}\mathcal{M}_{8}\otimes T\mathcal{M}_{8}\right)_{0}$.

\item A global section $\tilde{\phi}\in\Gamma\left(\Lambda^3 T\mathcal{M}_{8}\otimes T^{\ast}\mathcal{M}_{8}\right)_{0}$.

\end{itemize}

\noindent
We say then that $\mathfrak{S} = \left(g,\phi, \tilde{\phi}\right)$ is an intermediate structure on $\mathcal{M}_{8}$.
\end{definition}

\noindent
Notice that we do not require in principle the section $\phi$ to be everywhere non-vanishing nor $\mathcal{M}_{8}$ to be compact. This, together with the fact that any differentiable manifold can be equipped with a Riemannian metric means that there are only two obstructions for a manifold $\mathcal{M}_{8}$ to admit an intermediate structure $\mathfrak{S}$, namely

\begin{equation}
w_{1}\left(T\mathcal{M}_{8}\right) = 0\, , \qquad w_{2}\left(T\mathcal{M}_{8}\right) = 0\, ,
\end{equation}

\noindent
where $w_{1}$ and $w_{2}$ denote respectively the first and the second Stiefel-Whitney classes. These are precisely the obstructions for a manifold to be orientable and spinnable. Therefore, every orientable and spinnable eight-dimensional manifold admits an intermediate structure $\mathfrak{S}$, which is completely fixed once we choose an orientation, a metric, and spin structure and the two global sections $\phi$ and $\tilde{\phi}$. Given that every intermediate manifold is equipped with a Riemannian structure, we will take the volume form to be the one induced by the Riemannian metric $g$.

The idea is to study now the geometry of  intermediate manifolds in order to give them a physical meaning and understand if they can play any meaningful role in String/M/F-theory compactifications. 

It turns out that an intermediate structure $\mathfrak{S}$ on a manifold $\mathcal{M}_{8}$ gives rise to many interesting different geometric situations; some of them will be studied in the next sections. More precisely:

\begin{itemize}
\item In section \ref{sec:nd} we will assume the existence of a special kind of intermediate structure, called   $G_2$-intermediate structure (definition \ref{def:ndinter}).
We consider two cases:
\begin{itemize}
\item The seven-dimensional distribution on $\mathcal{M}_{8}$ orthogonal to the vector field $v$ (assumed to be non-vanishing) is completely integrable. We study the foliation by seven-dimensional submanifolds $\mathcal{M}_{7}$ tangent to the distribution, and geometric structures on them. This is done in sections \ref {sec:sevenmanifolds}, \ref{sec:G2manifolds}
 and \ref{sec:G2leavesstructure}.
\item The vector field $v$ is the generator of a free $S^1$ action
on $\mathcal{M}_{8}$. We study geometric structures on the seven-dimensional quotient manifold  $\mathcal{M}_{8}/S^1$. This is done in section \ref{sec:G2manifoldscircle}.
\end{itemize}
\item In section \ref{sec:generalintermediate} we consider 
general intermediate structures on $\mathcal{M}_{8}$.
 \end{itemize}
 
\noindent
We remark that most of the geometric conclusions we draw in sections
\ref{sec:nd} and \ref{sec:generalintermediate} do not make use of the fact that $\left(\mathcal{M}_{8},g\right)$ is a spin manifold.


\section{$G_2$-intermediate structures}
\label{sec:nd}


Let  $\mathcal{M}_{8}$ be an eight-dimensional oriented manifold with an intermediate structure $\mathfrak{S} = \left(g,\phi, \tilde{\phi}\right)$. In this section, we will use $\mathfrak{S}$ to embed seven-dimensional $G_2$ structure manifolds in $\mathcal{M}_{8}$  as the leaves of a foliation. Interestingly enough, the eight-dimensional manifold $\mathcal{M}_{8}$ can then be proven, under additional assumptions that we shall enumerate, to be a manifold with a topological $Spin\left(7\right)$-structure. This points to the interpretation of $\mathcal{M}_{8}$ as a plausible internal manifold in F-theory if it is also elliptically fibered, which, as we will see in section \ref{sec:generalintermediate}, can be the case. Since $\mathcal{M}_{8}$ also contains in a natural way a family of seven-dimensional manifolds with $G_{2}$ structure, which are the leaves of the foliation, we find a natural correspondence between the $G_{2}$-structure manifolds and the corresponding $Spin\left(7\right)$-structure manifolds, which may have a physical meaning in terms of dualities in String/M/F-theory. We leave the complete study of the holonomy of the different manifolds that appear in this set-up to a separate publication \cite{holonomyMCM}. 
 
Here, we will take $\tilde \phi$ to be the dual of $\phi$ by the musical isomorphisms $\flat$ and $\sharp$ given by $g$.\footnote{In terms of the supersymmetry spinors $\eta_{1}$ and $\eta_{2}$, this is \emph{equivalent} to taking them to be everywhere parallel.} Let us choose $\phi$ as follows:
\begin{equation}
\label{eq:phi3}
\phi = \phi_{3}\otimes v\in\Gamma\left(\Lambda^3 T^{\ast}\mathcal{M}_{8}\otimes T\mathcal{M}_{8}\right)_0\, ,
\end{equation}

\noindent
where $\phi_{3}\in\Gamma\left(\Lambda^3 T^{\ast}\mathcal{M}_{8}\right)$ is a three-form and $v\in\Gamma\left(T\mathcal{M}_{8}\right)$ is a vector field. The traceless condition on $\phi$ reads

\begin{equation}
\label{eq:phidegenerate}
\iota_{v}\phi_{3} = 0\, ,
\end{equation}

\noindent
and thus we conclude that $\phi_{3}$ is a degenerate three-form in $\mathcal{M}_{8}$. If $v$ is no-where vanishing, we can without loss of generality, take it to have unit norm with respect to the metric $g$ by a simple rescaling at every point $p\in\mathcal{M}_{8}$, i.e. by replacing $v_p$ with $\frac{v_{p}}{\sqrt{g(v_{p},v_{p})}}$. In this section we are going to consider that this is the case and hence, in the following, we will take $v$ to have unit norm:

\begin{equation} 
\label{unitnorm}
||v_{p}||=1 \text{ for all } p\in\mathcal{M}_{8}\, .
\end{equation}

\noindent
Most of the definitions and results will carry over to the case where $v$ has isolated zeros. We will come back to this point in section \ref{sec:singularities}. 

It is also natural to assume that
 
\begin{equation}
\label{eq:Kerphi3}
\mathrm{Ker}(\phi_{3\, p}) = \left\{\lambda v_{p}:\lambda \in \mathbb{R}\right\}\, , \quad p\in\mathcal{M}_{8}\, .
\end{equation}

\noindent
In other words, $f\cdot v\in\mathfrak{X}\left(\mathcal{M}_{8}\right)\, ,$ where $f\in C^{\infty}\left(\mathcal{M}_{8}\right)$ is any function, is the only vector field such that (\ref{eq:phidegenerate}) holds. Condition (\ref{eq:Kerphi3}) is the most natural situation compatible with equation (\ref{eq:phidegenerate}), since otherwise we would artificially introduce new vector fields that would span the kernel of $\phi_{3}$, and that are not incorporated in the intermediate structure $\mathfrak{S}$.

\noindent
We will call ``non-degenerate'' the intermediate structures that satisfy the properties mentioned so far. The name will be justified in proposition \ref{prop:phinonsingular}.
 
\begin{definition}\label{def:nondeg}
An \emph{non-degenerate} intermediate structure is one of the form $\left( g, \phi_{3}\otimes v\right)$ with $||v||=1$ and
  $ker(\phi_{3})=\mathbb{R}v$.
\end{definition}

\noindent
We define now, at every point $p\in\mathcal{M}_{8}$, a subset of the tangent space $T_{p}\mathcal{M}_{8}$ as follows:

\begin{equation}
H_{7\, p} = \left\{w_{p}\in T_{8\, p}\,|\, \xi_{p}\left(w_{p}\right) = 0\right\}\, ,\quad  \xi_p\equiv v_p^{\flat} \, \quad p\in\mathcal{M}_{8}\, . 
\end{equation}

\noindent
Notice that $H_{7\, p}$ is a well defined seven-dimensional vector subspace of $T_{8\, p}$ at every point $p\in\mathcal{M}_{8}$. 
In other words,
\begin{equation}
\label{eq:distribution1}
H_{7} = \left\{H_{7\, p}\, ,\,\, p\in\mathcal{M}_{8}\right\}\, ,
\end{equation}

\noindent
is a globally defined, codimension one, smooth distribution on $\mathcal{M}_{8}$, given by the kernel of the one form $\xi$. Notice that $H_{7}$  is simply the distribution of vectors orthogonal to $v$.

\begin{remark}
From a  {non-degenerate} intermediate structure $\left( g, \phi_{3}\otimes v\right)$, fixing a point $p$ and restricting the three-form
to   ${H_7}_p$,
 we obtain an element $\phi_3|_{{H_7}_p}\in  \Lambda^3 ({H_7}_p)^*$.
Doing so we do not lose information:  the unique extension of this element of $\Lambda^3 ({H_7}_p)^*$ to an element of $\Lambda^3 {T_p}^*\mathcal{M}_{8}$ that annihilates $v$, is exactly $(\phi_3)|_p$.
\end{remark}

We will be interested in a special case of the above, that gives rise to positive three forms: 

\begin{definition}\label{def:ndinter}
A {non-degenerate} intermediate structure  $\mathfrak{S}=\left( g,\phi_{3}\otimes v \right)$ 
on $\mathcal{M}_{8}$ is said to be a \emph{ $G_{2}$-intermediate structure} if, 
for all points $p$, $\phi_{3}|_{{H_7}_p}$ is a positive\footnote{Here ${H_7}_p$ has  the orientation given by the one of $\mathcal{M}_{8}$ and by $v$.} three-form, whose corresponding  metric is the restriction of $g$ to ${H_7}_p$. 
In this case, we say that $\left(\mathcal{M}_{8},\mathfrak{S}=\left( g,\phi_{3}\otimes v \right)\right)$  is a $G_{2}$-intermediate manifold.
\end{definition}

\noindent
\begin{remark}
A $G_{2}$-intermediate structure on $\mathcal{M}_{8}$ implies a reduction of the structure group of the manifold from $SO(8)$ to $G_2$, where $G_2$ is embedded in $SO(8)$ as
$$\left\{\left(\begin{array}{c|c}1 & 0 \\\hline 0 & A\end{array}\right):
A\in G_2\subset SO(7)\right\}.$$
More precisely, the $G_2$ reduction consists of  
$$\cup_{p\in \mathcal{M}_{8}}\{f: T_{p}\mathcal{M}_{8}\to\mathbb{R}^{8} \text{ oriented isometry such that $f(v)=e_0$ and $(f|_{{H_7}p})^*\phi_0=\phi_{3}|_{{H_7}p}$},\}$$
where $(e_0,\dots,e_7)$ is the canonical basis of $\mathbb{R}^{8}$ and $\phi_0$ as in definition \ref{def:phi0}.
\end{remark}
 
\noindent
We now proceed to study various instances of $G_2$-intermediate structures, as outlined at the end of section \ref{sec:8manifold}.

\subsection{A foliation of  $\mathcal{M}_{8}$ by  seven-manifolds}
\label{sec:sevenmanifolds}

Let  $\left( g, \phi_{3}\otimes v\right)$ be a  {non-degenerate} intermediate structure.
 We want to know under which conditions the distribution  $H_7$ is completely integrable. That is, we want to know under which conditions it defines a foliation of $\mathcal{M}_{8}$ such that the tangent space of the leaf  passing through $p\in\mathcal{M}_{8}$ is given by $H_{7\, p}$. The following proposition answers this question.

\begin{prop}
\label{prop:integrationxi}
Let $\mathcal{M}$ be a differentiable manifold equipped with a non-singular one form $\xi\in\Omega^{1}\left(\mathcal{M}\right)$. Then, the following conditions are equivalent

\begin{itemize}

\item $\xi\left(\left[w_1,w_2\right]\right) = 0$ for all $w_1,w_2\in\mathfrak{X}\left(\mathcal{M}\right)$ such that $\xi\left( w_1\right) = 0$ and $\xi\left( w_2\right) = 0$.

\item $\xi\wedge d\xi = 0$.

\item $d\xi = \beta\wedge\xi$ for some $\beta\in\Omega^{1}\left(\mathcal{M}\right)$.

\end{itemize}

\end{prop}

\begin{proof} See proposition 2.1 in \cite{foliations}.
\end{proof} 

\noindent
Therefore, using proposition \ref{prop:integrationxi} and the Frobrenius theorem we see that $H_{7}$ is completely integrable if and only if

\begin{equation}
\label{eq:integralxi}
\xi\wedge d\xi = 0,\,\;  \text{or equivalently},\,\;    d\xi = \beta\wedge\xi\, \text{for some $\beta\in\Omega^{1}\left(\mathcal{M}\right)$}.
\end{equation} 

\noindent If the distribution is completely integrable, then there exists a foliation, which we will denote by $\mathcal{F}_{\xi}$, of $\mathcal{M}_{8}$ by seven-dimensional leaves, whose tangent space at $p\in\mathcal{M}_{8}$ is given by $H_{7\, p}$. In other words: 
 assuming (\ref{eq:integralxi}),  at each $p\in\mathcal{M}_{8}$ there exist a   seven-dimensional submanifold $\mathcal{M}_{7\, p}\subset \mathcal{M}_{8}$ passing through $p$ and tangent to $H_7$.

\begin{remark}
Even when $v$ is compatible with the Riemannian metric $g$, in the sense the $v$ is a Killing vector field (one whose flow consists of isometries),  the orthogonal distribution $H_{7}$ may not be completely integrable. For instance, consider $\mathcal{M}_{8}=S^3\times \mathbb{R}^5$. The vector field $v$ generating the action of $S^1$ on the 3-sphere $S^3$ (obtained restricting the action on $\mathbb{C}^2$ by simultaneous rotations)  
is a Killing vector field, but $H_7$ is not  integrable. When 
$H_{7}$ happens to be completely integrable, its leaves are 
totally geodesic submanifolds\footnote{This means that if a geodesic of 
$\mathcal{M}_{8}$ is tangent to a leaf $L$, then it is contained in $L$ at all times.}.

More generally, replacing the one-dimensional orbits of a Killing vector field $v$ by submanifolds of higher dimension, we have the following. Let $(M,g)$ be a complete Riemannian manifold. A foliation $F$ on $M$
is called \emph{Riemannian foliation} \cite{Molino} if its holonomy 
acts by isometries on $H$, where denotes $H$ the distribution on $M$ orthogonal to  the
tangent spaces to the leaves of the foliation. 
 The distribution 
$H$ is not necessary integrable. When it is, one says that $F$
is a  \emph{polar foliation} (see for instance \cite{ThorbSRF}), and the leaves of $H$
are totally geodesic submanifolds.
\end{remark}

\noindent
Note that the one-form $\beta$ is not uniquely defined. Indeed, let $\beta$ and $\beta^{\prime}$ be two one-forms such that
$d\xi = \beta\wedge\xi = \beta^{\prime}\wedge\xi\,$.
Then 
\begin{equation}
\left(\beta - \beta^{\prime}\right)\wedge\xi = 0\,\,\Rightarrow\,\, \beta -\beta^{\prime} = f \xi\, ,
\end{equation}

\noindent
for some function $f\in C^{\infty}\left(\mathcal{M}\right)$.

\noindent
We now go back to our setting of an intermediate manifold  manifold $\mathcal{M}_{8}$. In that case, it is possible to write $\beta$ in terms of $v$.

\begin{prop}\label{prop:lvxi}
Let $\mathcal{M}$ be a differentiable manifold equipped with a Riemmanian metric $g$ and a vector field $v$. Under the assumption that $\xi = v^{\flat}$ satisfies the conditions of proposition \ref{prop:integrationxi}, we have:

a) the one-form $-\mathcal{L}_{v}\xi$, where $\mathcal{L}_{v}$ is the Lie derivative along $v$, is a proper  choice for $\beta$, that is, $d\xi = -\mathcal{L}_{v}\xi\wedge\xi$.

b) $d\xi=0$ if and only if $\mathcal{L}_{v}\xi=0$.
\end{prop}
\noindent
Notice that if  $v$ is a Killing vector field (i.e. $\mathcal{L}_{v}g=0$) then since $\mathcal{L}_{v}v=[v,v]=0$ we have
 $\mathcal{L}_{v}\xi=0$.

\begin{proof} For a), see proposition 2.2 in \cite{foliations}.

For b), notice that we have 
\begin{equation}
d\xi=0 \Leftrightarrow \iota_{v}d\xi=0
\end{equation}
 using  the formula $d\xi(w_1,w_2)=w_1(\xi(w_2))-w_2(\xi(w_1))
- \xi([w_1,w_2])$ together with the first condition in  proposition \ref{prop:integrationxi}. 
Further we have \begin{equation}
\mathcal{L}_{v}\xi = \iota_{v} d\xi + d\iota_{v} \xi = \iota_{v} d\xi
\end{equation}
 using Cartan's identity and the fact that $\iota_{v}\xi = 1$.
\end{proof}  
\noindent

\begin{remark}
The vector field $\left( -\mathcal{L}_{v}\xi\right)^{\sharp} $ is orthogonal to $v$. Indeed, the last computation in the proof of proposition \ref{prop:lvxi} shows that  $\mathcal{L}_{v}\xi =   \iota_{v} d\xi$, and therefore 
\begin{equation}
\iota_{v}\left(\mathcal{L}_{v}\xi\right) = 0\, .
\end{equation}
Consequently, if $H_7:=ker(\xi)$ defines a completely integrable distribution, then $\left(\mathcal{L}_{v}\xi\right)^{\sharp}$ is a vector field tangent to the leaves of the foliation.
\end{remark}

\begin{cor}
Let us assume that $H_7:=ker(\xi)$ defines a completely integrable distribution. Then the   foliation $\mathcal{F}_{\xi}$ is transversely orientable.
\end{cor}

\begin{proof} By definition, a foliation is transversely orientable if the normal bundle to the foliation is orientable. For a codimension one foliation 
this means exactly that there exist a vector field transverse to the foliation. This is the case in our situation, for  $v$ is a globally defined vector field transverse to $\mathcal{F}_{\xi}$.
\end{proof}

\begin{prop}
\label{prop:phinonsingular}
Let $\mathcal{M}_{8}$ be an eight-dimensional oriented manifold equipped with a non-degenerate intermediate structure  $\mathfrak{S}$ such that equation  (\ref{eq:integralxi}) holds. By proposition \ref{prop:integrationxi}, the kernel of $\xi\in\Omega^{1}\left(\mathcal{M}_{8}\right)$ defines an integrable distribution (\ref{eq:distribution1}), whose integral manifold passing through any point $p\in\mathcal{M}_{8}$ is denoted by $\mathcal{M}_{7\, p}\subset \mathcal{M}_{8}$. If we denote by 
\begin{equation}
i: \mathcal{M}_{7\, p}\hookrightarrow \mathcal{M}_{8}
\end{equation}
\noindent
the natural inclusion, then $i^{\ast}\phi_{3}$ is a non-degenerate  three-form on $\mathcal{M}_{7\, p}$.

\end{prop}

\begin{proof} Fix $p\in\mathcal{M}_{8}$ and $w\in T_p\mathcal{M}_{7, p}$.
Suppose that $\iota_{w} (i^{\ast}\phi_{3})=0$, i.e. that $i^{\ast}(\iota_{w}\phi_{3})=0$.
 Since $v_p\in Ker(\phi_{3})$ and $T_p\mathcal{M}_{8, p}=T_p\mathcal{M}_{7, p}\oplus \mathbb{R}v_p$, we conclude that $\iota_{w} \phi_{3}=0$.
 Equation \eqref{eq:Kerphi3} implies that $w$ is a multiple of $v_p$. Since $w$ is tangent to ${M}_{7, p}$ while $v_p$ is transverse to it, we obtain $w=0$.
\end{proof} 

\noindent
To summarize, assuming (\ref{eq:integralxi}), $\mathcal{M}_{8}$ is foliated by seven-dimensional manifolds which are equipped with a non-degenerate three-form $i^{\ast}\phi_{3}$. 

\subsection{$G_{2}$-structure seven-manifolds in $\mathcal{M}_{8}$}
\label{sec:G2manifolds}

We are interested in connecting the seven-dimensional manifolds $\left\{\mathcal{M}_{7\, p}\, , \, p\in\mathcal{M}_{8}\right\}$ that form the leaves of the foliation to the seven dimensional manifolds that appear as internal spaces in off-shell $\mathcal{N}=1$ M-theory compactifications, which have $G_{2}$-structure. Therefore, it is natural to ask whether the three-form $i^{\ast}\phi_{3}$ that exists on every leaf may define a $G_{2}$-structure on the leaves. In other words, we want to know if the three forms $i^{\ast}\phi_{3}$ can be taken to be positive. Remarkably enough, since every positive form is non-degenerate, the previous construction is consistent with taking $\phi_{3}$ in such a way that $i^{\ast}\phi_{3}$ is a positive form. Of course, what we want to know is, given a vector field $v$ on $\mathcal{M}_{8}$, if there is any obstruction to choose $\phi_{3}$ in such a way that $i^{\ast}\phi_{3}$ is a positive three-form. The following proposition answers this question.

\begin{prop}
\label{prop:G2leaves}
Let $\mathcal{M}_{8}$ be an eight-dimensional oriented Riemannian spin manifold equipped with a transversely orientable, codimension one foliation $\mathcal{F}$. Then each leaf of the foliation admits a $G_{2}$-structure,  with the property that
  its associated
Riemannian metric is the pullback of the given Riemannian metric on $\mathcal{M}_{8}$.
\end{prop}

\begin{proof}
Since $\mathcal{F}$ is transversely orientable, the normal bundle $\mathcal{N}$ to $\mathcal{F}$ is trivial, and in particular orientable, that is, $w_{1}\left(\mathcal{N}\right) = 0$. In addition, $w_{2}\left(\mathcal{N}\right) = 0$ since $2>\mathrm{rank}\, \mathcal{N}$. Now, using \cite[Ch.II, proposition 2.15]{Spingeometry}
 we deduce that the leaves of $\mathcal{F}$ are spin manifolds. We finally conclude from proposition \ref{prop:topologicalreductions7noncompact} (and its proof) that on each leaf there exists topological $G_{2}$-structure with the above property.\end{proof} 

\noindent
Notice that proposition \ref{prop:G2leaves} is a mere existence result, and does not address the issue of whether
the $G_{2}$-structures vary smoothly\footnote{An example where this happens is the following. Let $N$ be an oriented seven-dimensional manifold endowed with a positive three-form $\Phi^0$, and $f\colon N\to N$ a diffeomorphism satisfying $f^*\Phi^0=\Phi^0$. Take the mapping torus $\mathcal{M}_{8}:= ([0,1]\times N)/\sim$, where the equivalence relation $\sim$ identifies $(0,p)$ with $(1,f(p))$ for all $p\in N$, together with the codimension one foliation given by the fibers of the projection $\mathcal{M}_{8}\to [0,1]/(0\sim 1)=S^1$. Notice that all fibers are diffeomorphic to $N$.} from leaf to leaf. When the $G_{2}$-structures do vary smoothly, we have the following. 
 
\begin{prop}
\label{prop:smooth}
Let $\mathcal{M}_{8}$ be an eight-dimensional oriented Riemannian spin manifold, equipped with a transversely orientable, codimension one foliation $\mathcal{F}$. Let $v$ be a unit vector field orthogonal to $\mathcal{F}$. 

If the positive three-forms on the leaves of  $\mathcal{F}$ 
mentioned in proposition \ref{prop:G2leaves} can be chosen so 
that they vary smoothly from leaf to leaf,  then $\mathcal{M}_{8}$ is equipped with a unique three-form $\phi_{3}\in\Omega^{3}\left(\mathcal{M}_{8}\right)$ such that $\iota_{v}\phi_{3} = 0$ and 
$\phi_3$ pulls back to the given positive three-forms on all leaves.
In other words, $\mathfrak{S}=\left( g,\phi_{3}\otimes v \right)$ is a $G_2$-intermediate structure on $\mathcal{M}_{8}$ such
that $H_7$ is integrable.
\end{prop}

\begin{remark}
It is not always possible to choose in a globally smooth way the three-forms on the leaves  mentioned in proposition \ref{prop:G2leaves}. Indeed, this is possible if{f} there exists 
a topological $Spin(7)$-structure on $\mathcal{M}_{8}$ whose associated Riemannian metric is the given one on $\mathcal{M}_{8}$.
This follows immediately from Thm. \ref{thm:Spin(7)fromG2} and Thm. \ref{prop:Spin(7)toG2}.
 
A conceptual explanation for the above failure is the following.
The $G_2$-structures on the leaves may not be unique topologically (i.e. up to continuous deformation). If the bundle of deformation classes (a discrete bundle over one dimensional the leaf space)
has monodromy,  it is not possible to choose the three-forms in a continuous way. We thank Dominic Joyce for pointing this out to us.
\end{remark}

\begin{proof}
We will prove the proposition by explicitly constructing the three-form $\phi_{3}$. 
Let us denote by $\mathcal{M}_{7\, p}$ the leaf of $\mathcal{F}$ passing through $p\in\mathcal{M}_{8}$.
Every leaf $\mathcal{M}_{7, p}$ is equipped with a positive three-form $\phi^{(0,p)}_{3}\in\Omega^{3}\left(\mathcal{M}_{7, p}\right)$. 
Let $(e_{1},\dots,e_{7})$ be any  basis of $T_{p}\mathcal{M}_{7, p}$. Now we define, for all $p\in\mathcal{M}_{8}$, $\phi_{3}|_{p}\in\Lambda^3 T^{\ast}_{p}\mathcal{M}_{8}$ as follows:
\begin{equation}
\label{eq:phi3M8}
\phi_{3}|_{p}\left(e_{i}, e_{j}, e_{k}\right) =  \phi^{(0,p)}_{3}\left(e_{i}, e_{j}, e_{k}\right)\, ,  \quad \phi_{3}|_{p}\left(e_{i}, e_{j}, v\right) = 0\, ,\quad i, j, k = 1,\dots, 7\,.
\end{equation}
\noindent
(Notice that $\phi_{3}|_{p}$ is independent of the choice of basis).
By construction, $\phi_{3}$ defined as in equation (\ref{eq:phi3M8}) is positive on the leaves, and since the positive forms $\{\phi^{(0,p)}_{3}\}$ vary smoothly from leaf to leaf, 
 $\phi_{3}$ is also smooth. Furthermore, by construction,  $\iota_{v}\phi_{3} = 0$. We conclude that $\phi_{3}$ is the desired three-form.
\end{proof}

Notice that the seven-dimensional manifolds do not depend on $\phi_{3}$ but only on $v$, and therefore they do not react to any changes in $\phi_{3}$. However, taking $\phi_{3}$ to be the three-form that induces a $G_{2}$-structure on the leaves is the sensible choice that allows to make contact with the rank eight bundle (\ref{decomposition_gl7_V}) introduced in \cite{Grana:2012zn}.  According to the decomposition (\ref{decomposition_gl7_V}), the complex four-form (\ref{eq:phi4}) that is constructed from the complex Weyl $Spin(8)$ spinor $\eta$  decomposes as in equation (\ref{eq:Omegabundledecomposition}). Due to the   theorem \ref{thm:Spin(7)fromG2} below, we precisely obtain an analogous decomposition here but now the rank eight vector bundle is the tangent bundle of an actual manifold $\mathcal{M}_{8}$, which is foliated by seven-dimensional $G_{2}$-manifolds. The tangent space of the leaves is a rank seven vector bundle analogous to $T_{7}$ in equation (\ref{eq:Omegabundledecomposition}). Furthermore, theorem \ref{thm:Spin(7)fromG2}  shows that ${\cal M}_8$ has a Spin(7)-structure. This theorem   
is a mild generalization of \cite[Prop. 11.4.10]{Joyce2007} (we reproduce the latter in the appendix as proposition \ref{prop:G2T2}).

\begin{thm}
\label{thm:Spin(7)fromG2} 
Let $\left(\mathcal{M}_{8},\mathfrak{S}=\left( g,\phi_{3}\otimes v \right)\right)$ be a $G_{2}$-intermediate manifold. Then, $\mathcal{M}_{8}$ has a $Spin(7)$-structure defined by the admissible four-form

\begin{equation}
\label{eq:Omegafromphi3}
\Omega = v^{\flat}\wedge  \phi_{3}+ \ast\left(v^{\flat}\wedge  \phi_{3}\right)\, .
\end{equation}
The Riemannian metric associated to $\Omega$ is exactly $g$.
\end{thm}

\begin{proof}
Fix $p\in \mathcal{M}_{8}$. Applying lemma \ref{lem:phitoOmega} to the vector space $V:=T_p\mathcal{M}_{8}$ and subspace $H:=(H_7)_p$ 
shows that $\Omega|_p$ is an admissible form. We conclude that 
$\Omega$ is an admissible form on $\mathcal{M}_{8}$.
The final statement follows from  remark \ref{rem:innonV}.
\end{proof}

\noindent
The following    is a converse to theorem \ref{thm:Spin(7)fromG2}. 
 
\begin{thm}
\label{prop:Spin(7)toG2}
Suppose  $\mathcal{M}_{8}$ is an 
 oriented Riemannian spin manifold with 
a topological $Spin(7)$-structure (whose associated metric is the given one $g$), and denote by  $\Omega$ the corresponding admissible four-form.  
Let $\mathcal{F}$ be a transversely orientable, codimension one foliation on $\mathcal{M}_{8}$,
and let $v$ be a  vector field orthogonal to $\mathcal{F}$ with $||v||=1$.

Then, for any leaf $i \colon L\hookrightarrow \mathcal{M}_{8}$ of $\mathcal{F}$, the three-form $i^*(\iota_v\Omega)$ is positive,   varies smoothly from leaf to leaf, and
 the metric on $L$ associated to this positive three-form is the restriction of $g$. (In other words:  $\mathfrak{S}=\left( g,\iota_v\Omega\otimes v \right)$ is a $G_2$-intermediate structure on $\mathcal{M}_{8}$ such that $H_7$ is integrable.)
\end{thm}

\begin{proof} Fix a leaf $L$.
 Applying lemma \ref{lem:anyv} to the vector space $V:=T_p\mathcal{M}_{8}$ and subspace $H:=(H_7)_p$, for all   $p\in L$, shows
 that $i^*(\iota_v\Omega)$ is a positive three-form on $L$.
The statement on the metric follows from Remark \ref{rem:metricH}. 
\end{proof}

For the sake of clarity, we describe in coordinates the four-form $\Omega$ obtained in theorem \ref{thm:Spin(7)fromG2}.
Fix $p\in\mathcal{M}_{8}$ and denote by $\mathcal{M}_{7\, p}$ the leaf of $H_7$ through $p$. There exists an orientation-preserving coordinate system $\left\{\mathcal{U}, {\bf x}_{8} = (x_{1},\dots,x_{8})\right\}$ about $p$ such that $v_p=\frac{\partial}{\partial x_{1}}|_{p}$ and
\begin{equation}
\left.\left(\frac{\partial}{\partial x_{1}},\dots,\frac{\partial}{\partial x_{8}}\right)\right|_{p}\, 
\end{equation}
\noindent
is an orthonormal basis of $T_{p}\mathcal{M}_{8}$.
 The metric $g$ at $p$ then can be written as
\begin{equation}
\label{eq:canmetriccoord}
g|_{p} = dx^{2}_{1}|_{p}+\cdots + dx^{2}_{8}|_{p}\,. 
\end{equation}
The three-form $\phi_3$ on $\mathcal{M}_{8}$ at $p$ has coordinate expression
\begin{equation} (\phi_3)|_p=\sum_{2\le i<j<k\le 8} h_{ijk}(dx_{i}|_p)\wedge (dx_{j}|_p)
\wedge (dx_{k}|_p)
\end{equation}
where $h_{ijk}$ are real numbers.
\noindent
On the other hand, ${\bf x}_{7} = \left(x_{2},\dots,x_{8}\right)$ is a coordinate system on $\mathcal{M}_{7\, p}$ around $p$ such that
\begin{equation}
\label{eq:orthonormal7}
\left.\left(\frac{\partial}{\partial x_{2}},\dots,\frac{\partial}{\partial x_{8}}\right)\right|_{p}\, 
\end{equation}
\noindent
is an orthonormal basis of $T_{p}\mathcal{M}_{7\, p}$. In particular,
\begin{equation}
dx^{2}_{2}|_{p}+\cdots + dx^{2}_{8}|_{p}\, ,
\end{equation}
\noindent
is the restricted metric 
on $\mathcal{M}_{7\, p}$ in the coordinate system (\ref{eq:orthonormal7}). Using that $v_p^{\flat} = dx_{1}|_p$ we finally obtain
\begin{equation} 
\Omega = dx_{1}|_p\wedge  \phi_{3}|_p + \ast\left(dx_{1}|_p\wedge  \phi_{3}\right)\,, 
\end{equation}
where the Hodge dual is compute w.r.t. the standard metric \eqref{eq:canmetriccoord}.

\begin{remark}\label{rem:noint}
Notice that Thm. \ref{thm:Spin(7)fromG2} does not assume that $H_7$ be integrable,and Thm. \ref{prop:Spin(7)toG2}
 does not use the integrability in an essential way either. The analogue of Thm. \ref{prop:Spin(7)toG2}   without the integrability assumption on $H_7$ reads:
 
 Suppose  $\mathcal{M}_{8}$ is an 
 oriented Riemannian spin manifold with 
a topological $Spin(7)$-structure (whose associated metric is the given one $g$), and denote by  $\Omega$ the corresponding admissible four-form.  
Let $H_7$ be transversely orientable, codimension one smooth distribution on $\mathcal{M}_{8}$,
and let $v$ be a  vector field orthogonal to $\mathcal{F}$ with $||v||=1$.
Then, at every $p\in \mathcal{M}_{8}$, the three-form $(\iota_v\Omega)|_{{H_7}_p}$ is positive,  varies smoothly with the point $p$, and  the metric on ${{H_7}_p}$ associated to this positive three-form is the restriction of $g$. (In other words:  $\mathfrak{S}=\left( g,\iota_v\Omega\otimes v \right)$ is a $G_2$-intermediate structure on $\mathcal{M}_{8}$.)
\end{remark}



\begin{remark}
Let $(\mathcal{M}_{8},g)$ be an   eight-dimensional oriented Riemannian spin manifold. Let $\mathcal{F}$ be a codimension-one, transversely orientable foliation, and denote by $v$ a unit vector field orthogonal to the foliation. Denote by $i^{\ast}g$ the restricted inner product on the each leaf $i\colon L\hookrightarrow \mathcal{F}$. Proposition \ref{prop:bijection} immediately implies that there is a bijection
\begin{align*}
 \{\text{admissible four-forms on $\mathcal{M}_{8}$ inducing $g$}\}
&\leftrightarrow \{\text{positive three-forms on the leaves of $\mathcal{F}$ inducing $i^*g$} \\
&\;\;\;\; \;\;\;\text{ and varying smoothly from leaf to leaf.}\}
\end{align*}
The purpose of theorem \ref{thm:Spin(7)fromG2}  and theorem \ref{prop:Spin(7)toG2}
above is to spell out this correspondence in terms of $Spin(7)$-structures and $G_2$-structures.

\end{remark}
 
\begin{cor}
\label{cor:admissible3manifold}
Let $(\mathcal{M}_{8},g)$ be an   eight-dimensional oriented Riemannian spin manifold. Let $\mathcal{F}$ be a codimension-one, transversely orientable, foliation, and denote by $v$ a unit vector field transverse to the foliation. Denote by $i^{\ast}g$ the restricted inner product on   each leaf $i\colon L\hookrightarrow \mathcal{F}$.

\begin{enumerate}

\item If $ v^{\flat}\wedge  \phi+ \ast\left(v^{\flat}\wedge  \phi\right)$ is an admissible four-form on $\mathcal{M}_{8}$ (not necessarily inducing the inner product $g$), then 
$i^{\ast}\phi$ is a positive three-form on every leaf $i: L\hookrightarrow\mathcal{M}_{8}$, varying smoothly from leaf to leaf.

\item $ v^{\flat}\wedge  \phi+ \ast\left(v^{\flat}\wedge  \phi\right)$ is an admissible four-form on $\mathcal{M}_{8}$ inducing $g$ if and only if $i^{\ast}\phi$ is a positive and smooth three-form on $\mathcal{F}$ inducing $i^{\ast}g$.
\end{enumerate}
\end{cor}
\begin{proof}
Apply corollary \ref{cor:admissible3} to the vector space $V:=T_p\mathcal{M}_{8}$ and subspace $H:=(H_7)_p$,
for all   $p\in \mathcal{M}_{8}$.
\end{proof}

Note that equation (\ref{eq:Omegafromphi3}) defining a real admissible four-form is of the form (\ref{eq:Omegabundledecomposition}), constructed in terms of a complex spinor $\eta = \eta_{1} + i\eta_{2}$. The only difference is that equation (\ref{eq:Omegabundledecomposition}) is obtained by a $7+1$ split of a rank-eight vector bundle, while equation (\ref{eq:Omegafromphi3}) is obtained from a seven-dimensional foliation of an eight-dimensional manifold. Therefore, equation (\ref{eq:Omegafromphi3}) can be thought as the \emph{geometrization} of equation (\ref{eq:Omegabundledecomposition}) in terms of manifolds and   tangent spaces, instead of a rank-eight vector bundle which does not correspond to the tangent space of any manifold.  We can connect both set-ups by writing $\Omega$ in terms of a real Weyl spinor $\chi\in\Gamma\left(\mathbb{S}^{\mathbb{R}+}_{8}\right)$, which exists since $\mathcal{M}_{8}$ has $Spin(7)$ structure: 

\begin{equation}
\label{eq:Omegachi}
\Omega = \chi^{T}\gamma_{(4)}\chi = v^{\flat}\wedge\phi_{3} + \ast\left( v^{\flat}\wedge\phi_{3}\right)\,. 
\end{equation} 

\noindent
However $\chi$ is not yet a \emph{supersymmetry spinor}, since it is a section the spin bundle $\mathbb{S}^{\mathbb{R}+}_{8}\to\mathcal{M}_{8}$ over $\mathcal{M}_{8}$, while  $\eta_{a}\, , a=1,2\, ,$ are sections of the spin bundle $\mathbb{S}^{\mathbb{R}}_{7}\to\mathcal{M}_{7}$. As we have explained, $\eta_{a}\, , a=1,2\, ,$ are $Spin(7)$ real spinors, which can be identified with real $Spin(8)$-Weyl spinors. As elements of a vector space, the previous identification is perfectly fine, but if we carry it out globally we obtain that $\eta_{a}$ is now a section of $\mathbb{S}^{\mathbb{R}+}_{8}\to\mathcal{M}_{7}$, that is, the base manifold is $\mathcal{M}_{7}$ instead of $\mathcal{M}_{8}$. Therefore, given the geometric structure defined on an intermediate manifold $\left(\mathcal{M}_{8},\mathfrak{S}\right)$, it is natural to identify one of the $\eta_{a}$ (let's say $\eta_1$)  with the pull-black of the real, Weyl spinor $\chi$ over $\mathcal{M}_{8}$, i.e.

\begin{equation}
\label{eq:chieta}
i^{\ast}\chi = \eta_{1}\, ,  
\end{equation}

\noindent
and take $\eta_{2} = 0$. Notice that $i^{\ast}\chi$ denotes a family of spinors, one spinor for each leaf, that is, a section of the vector bundle $\mathbb{S}^{\mathbb{R}+}_{8}\to\mathcal{M}_{7}$, where we denote by $\mathcal{M}_{7}$ a generic leaf. We have taken $\eta_{2} = 0$ in order to obtain a real four-form in equation (\ref{eq:phi4}), since for the particular intermediate structure that we are considering, we only obtain a real admissible four-form. Notice that we are implicitly identifying the $G_{2}$-structure leaves as internal spaces in M-theory compactifications down to four dimensions. We will see in section \ref{sec:generalintermediate} that when we do not consider $\tilde{\phi}$ to be the \emph{dual} of $\phi$ then we can obtain the real as well as the imaginary parts of the complex four-form constructed out of the complex spinor $\eta$, and therefore we obtain the general geometric situation as it appears from off-shell supersymmetry. Notice that (\ref{eq:chieta}) does not completely determine $\chi$ from $\eta_{1}$, as there are several different choices of $\chi$ such that equation (\ref{eq:chieta}) is satisfied. Loosely speaking, there is some sort of \emph{gauge freedom} in order to choose $\chi$. In order to fix the gauge, a natural condition is to impose
\begin{equation}
\label{eq:chiv}
\nabla_{v}\chi = 0\, .
\end{equation}

\noindent
Intuitively speaking, equation (\ref{eq:chiv}) is the \emph{covariant} version of the partial derivative being zero as a way of saying that there is no dependence on a particular coordinate, in this case the \emph{orthogonal} coordinate to the leaves, since $v$ is a globally defined coordinate in $\mathcal{M}_{8}$, perpendicular to the $G_{2}$-structure leaves that conform the foliation defined by $v$.


\subsection{Structure of the foliation}
\label{sec:G2leavesstructure}


In   section \ref{sec:sevenmanifolds} we have defined a codimension one distribution in $\mathcal{M}_{8}$ which is given by the kernel of the one-form $\xi$. The conditions under which this distribution is completely integrable are given in proposition \ref{prop:integrationxi}. In the case that any of those conditions is satisfied, the eight-dimensional manifold $\mathcal{M}_{8}$ is foliated by seven-dimensional manifolds $\mathcal{M}_{7\, p}\, ,\, p\in\mathcal{M}_{8}$. Each leaf $\mathcal{M}_{7\, p}$ admits a $G_{2}$- structure (see proposition \ref
{prop:G2leaves}).  
In some cases, as explained in section \ref{sec:G2manifolds}, 
the three-form $\phi_{3}$ can be chosen to give, using the pull-back of the natural inclusion $i: \mathcal{M}_{7\, p}\hookrightarrow\mathcal{M}_{8}$, the three-form that defines the $G_{2}$-structure on the leaves. The leaves however do not have to be neither diffeomorphic nor compact and in principle very little is known about them, aside from the fact that they are smooth seven-dimensional manifolds. The purpose of this section is thus to explore in some detail the geometry of the leaves, trying to characterize them as much as possible, in order to clarify the relation between the $G_{2}$-manifolds that are the leaves and the $Spin(7)$-manifold that is the total space $\mathcal{M}_{8}$. 

\begin{prop}
\label{prop:closedclosed}
Let $\mathcal{M}$ be a closed manifold\footnote{A closed manifold is a compact manifold without boundary.} equipped with a codimension one foliation $\mathcal{F}$ defined by a non-singular closed one-form $\xi$. Then there exists a transversal vector field, whose flow consists of diffeomorphisms preserving $\mathcal{F}$, \emph{i.e.}, mapping leaves into leaves.  
\end{prop} 

\begin{proof}See reference \cite{Reeb1}. 
\end{proof}

\begin{cor}
\label{cor:leavesdiff}
In the situation of proposition \ref{prop:closedclosed}, and assuming that $\mathcal{M}$ is connected, all the leaves are diffeomorphic. 
\end{cor}

\begin{prop}
\label{thm:cod1foliations}
Let $\mathcal{M}$ be a closed manifold equipped with a non-singular closed one-form $\xi$. Then, there exists a fibration $f:\mathcal{M}\to S^{1}$ over the circle. Moreover, let $\xi^{\prime} = f^{\ast} d\theta$. Then the fibration can be chosen such that $||\xi-\xi^{\prime} ||<\epsilon$, where $\epsilon>0$ is any prescribed number. 
\end{prop}

\begin{proof}See reference \cite{Tischler}.
\end{proof} 

\noindent
Here $||\cdot ||$ stands for the point-wise norm on forms, defined by a Riemannian metric on $\mathcal{M}$. What proposition \ref{thm:cod1foliations} states is that, if we denote by $\mathcal{F}$ 
 the foliation defined by the closed non-singular one-form $\xi$,
there  exists arbitrarily close (in the $C^0$-sense) foliations which are given by the fibers of a fibration $f: \mathcal{M}\to S^{1}$. 
 
\noindent
To summarize, if $\mathcal{M}$ is closed and $d\xi = 0$, then the behavior of the foliation is under control: all the leaves are diffeomorphic and in addition there is a fibration $f: \mathcal{M}\to S^{1}$ such that the fibres are arbitrarily close to the leaves of the foliation defined by $\xi$. A typical example of this given by the two-torus $S^1\times S^1$ with coordinates $\theta_1$ and $\theta_2$, and by $\xi=d\theta_1+\lambda d\theta_2$ for some real number $\lambda$: if $\lambda$ is a rational, $\mathcal{F}$ arises from a fibration, otherwise one can approximate $\mathcal{F}$ by taking the kernel of $d\theta_1+\lambda' d\theta_2$, where $\lambda'$ is a rational number very close to $\lambda$.

Proposition \ref{thm:cod1foliations} applies to closed manifolds, that is, compact manifolds without boundary. We want, however, to consider the situation where $\mathcal{M}_{8}$ may have a boundary, since in some M-theory/F-theory applications we expect to find manifolds with boundary. There is indeed a result concerning the case of compact manifolds with boundary.

\begin{prop}
\label{thm:cod1foliationsboundary} Let $\mathcal{F}$ be a codimension one, $C^{1}$, transversely oriented foliation on a compact manifold $\mathcal{M}$ (possibly with boundary) with a compact leaf $ {L}$ such that $H^{1}\left( {L},\mathbb{R}\right) = 0$. Then, all the leaves of $\mathcal{F}$ are diffeomorphic with $ {L}$, and the leaves of $\mathcal{F}$ are the fibers of a fibration of $\mathcal{M}$ over $S^{1}$ or the interval $I$. We assume here that the boundary of $\mathcal{M}$ is non-empty, then $\partial\mathcal{M}$ is a union of leaves of $\mathcal{F}$. 
\end{prop}

\begin{proof}See theorem 1 in reference \cite{Thurston}.
\end{proof}  
\noindent
In order to determine if $\mathcal{M}_{8}$ is an admissible internal space in F-theory, since we already know that in some cases, see theorem \ref{thm:Spin(7)fromG2}, it is a $Spin(7)$-structure manifold, we have to study if it is elliptically fibered. That is a difficult question to answer using just with the geometric data contained in the particular case of intermediate structure that we are considering. We will see in section \ref{sec:generalintermediate} that if we consider a more general intermediate structure then it is more natural to obtain that $\mathcal{M}_{8}$ is elliptically fibered. In any case, the following construction may shed some light on this issue.

\begin{itemize}

\item We start with a $G_{2}$-intermediate manifold $\left(\mathcal{M}_{8},\mathfrak{S}\right)$, where $\mathfrak{S}=\left( g, \phi_{3}\otimes v\right)$. Therefore by proposition \ref{prop:integrationxi}
 we know that the following holds:

\begin{equation}
\label{eq:xicondition}
d\xi = \beta\wedge\xi\,  \qquad \text{ for some }\beta\in \Omega^{1}\left(\mathcal{M}_{8}\right)\, ,
\end{equation}

\noindent
where $\xi = \iota_{v}g$. The properties of the foliation depend now exclusively on $\xi$ and we have thus two possibilities, namely

\begin{enumerate}

\item{ $d\xi = 0$.
In this case then, depending on the particular properties of $\mathcal{M}_{8}$ we can be able to apply corollary \ref{cor:leavesdiff} or propositions \ref{thm:cod1foliations} and \ref{thm:cod1foliationsboundary}. If $\mathcal{M}_{8}$ is taken to be a closed manifold, using proposition \ref{thm:cod1foliations}, we conclude that there exists a fibration over the circle

\begin{equation}
\label{eq:S1fibration8}
f:\mathcal{M}_{8}\to S^{1}\, ,
\end{equation}

\noindent
with seven-dimensional fibres which are arbitrarily close to the leaves of the foliation defined by $\xi$. We already concluded that the leaves are seven-dimensional manifolds of $G_{2}$-structure. However, we are interested in having a $G_{2}$-structure defined also on the fibres of the fibration (\ref{eq:S1fibration8}), which are the seven dimensional  internal manifolds appearing in off-shell $\mathcal{N}=1$ supersymmetric  M-theory compactifications. Indeed, the following proposition holds.

\begin{prop}
\label{prop:fibresG2}
Let $\mathcal{M}_{8}$ be an eight-dimensional oriented closed manifold equipped with an $G_{2}$-intermediate structure $\mathfrak{S} = \left( g, \phi_{3}\otimes v\right)$ such that $d\xi = 0$, where $\xi = v^{\flat}$. Then $\mathcal{M}_{8}$ admits a fibration over the circle $S^{1}$
\begin{equation}
f:\mathcal{M}_{8}\to S^{1}\, ,
\end{equation}

\noindent
whose fibres are seven dimensional manifolds with $G_{2}$-structure, varying smoothly from fiber to fiber. 
\end{prop}

\begin{proof} 
By proposition \ref{thm:cod1foliations}, for any  $\epsilon>0$ there exists a fibration $f:\mathcal{M}_{8}\to S^{1}$ over the circle such that $||\xi-\xi^{\prime} ||<\epsilon$, where $\xi^{\prime} := f^{\ast} d\theta$. Since $v$ is orthogonal to $ker(\xi)$, we  can choose $\epsilon$ so small that 
$v$ is transverse to $ker(\xi')$. By theorem \ref{thm:Spin(7)fromG2} there is an admissible four-form $\Omega$ on $\mathcal{M}_{8}$. Applying at every point 
proposition \ref{lem:anyv} to $\Omega$ we see that  $(\iota_v\Omega)|_{ker(\xi')}$ is a positive three-form. Thus we obtain a positive three-form on the fibers of $f$, which clearly varies smoothly from fiber to fiber.
\end{proof} 
\begin{remark}
In proposition \ref{prop:fibresG2}, the metric on the fibers associated to the positive three-forms on the fibers is usually not the pullback of the metric $g$ on   $\mathcal{M}_{8}$. 
\end{remark}

\item{ $d\xi$ no-where vanishing, which implies that $\beta$ is globally defined and nowhere vanishing. This means in particular that $\xi$ and $\beta$ are nowhere parallel.

Since $\beta$ and $\xi$ are never parallel, $i^{\ast}\beta$ is a non-singular, no-where vanishing one-form on every leaf $\mathcal{M}_{7\, p}\, ,\, p\in\mathcal{M}_{8}$ of the foliation. Using equation (\ref{eq:xicondition}) twice we see that  
\begin{equation}
d\beta\wedge\xi = d\beta\wedge\xi- \beta\wedge \underbrace{d\xi}_{\beta\wedge\xi} =d(\beta\wedge\xi)= d(d\xi)=
 0\, ,
\end{equation}

\noindent
and thus
\begin{equation}
d\beta = \gamma\wedge \xi\,\qquad \text{ for some } \gamma\in\Omega^{1}\left(\mathcal{M}_{8}\right)\, ,
\end{equation}

\noindent
which in turn implies
\begin{equation}
\label{eq:iastbeta}
di^{\ast}\beta = 0\, .
\end{equation}

\noindent
Therefore we conclude that the foliation defined by $i^*\beta$ naturally satisfies the conditions of proposition \ref{prop:closedclosed}, corollary \ref{cor:leavesdiff} and theorem \ref{thm:cod1foliations} on those leaves $\mathcal{M}_{7\, p}$ that are closed, if any. The situation is then the following: $\mathcal{M}_{8}$ is foliated by   $G_{2}$-structure manifolds, which in principle do not need to be diffeomorphic nor compact. Each leaf, that is, each $G_{2}$-structure manifold, can in turn be foliated by six dimensional leaves. Since (\ref{eq:iastbeta}) says that the one-form defining such foliation is automatically closed, we deduce that if the $G_{2}$-structure manifold is closed then it is a fibration over $S^{1}$. The global structure of the eight-dimensional manifolds $\mathcal{M}_{8}$ cannot be determined yet, as the discussion above allows for many different, more or less complicated, possibilities.}

\item{$d\xi$ not identically zero but neither no-where vanishing. Therefore, $d\xi$ can be zero at some points (or sets) in $\mathcal{M}_{8}$. Here we find a mixed situation. In those points belonging to the support of $d\xi$\footnote{Not taking the corresponding closure.}, $\beta$ is non-vanishing and not parallel to $\xi$, and therefore the discussion in point 2 applies. At the points where $d\xi$ is zero, $\beta$ is zero over the corresponding leave, and therefore no further foliation can be defined. Hence, we find that $\mathcal{M}_{8}$ is foliated by $G_{2}$-structure seven-dimensional manifolds and that some of them, determined by the points where $d\xi$ is non-zero, are in turn foliated in terms of six-dimensional manifolds. }
}
\end{enumerate}

\end{itemize}


\subsection{$\mathcal{M}_{8}$ as a $S^{1}$ principal bundle}
\label{sec:G2manifoldscircle}


In sections \ref {sec:sevenmanifolds}, \ref{sec:G2manifolds}
 and \ref{sec:G2leavesstructure} we considered the dimension seven distribution $H_7$  (see \eqref{eq:distribution1})
 of vectors perpendicular to the vector field $v$, and we assumed that it be integrable.  
In this section we explore another natural possibility, namely, the dimension one distribution

\begin{equation}
H_{1 p} = \left\{ w_{p}\in T_{p}\mathcal{M}_{8}\,\, |\,\, w_{p} = \lambda v_{p}\, , \quad \lambda \in \mathbb{R}\right\}\, , \quad p \in\mathcal{M}_{8}\, .
\end{equation}



\noindent
Note that the dimension one distribution $H_{1} = \left\{ H_{1 p}\, ,\,\, p\in \mathcal{M}_{8}\right\}$ is automatically integrable, the leaves being simply the integral curves of $v\in\mathfrak{X}\left( \mathcal{M}\right)$. 
  There is a one-to-one correspondence between complete vector fields on a differentiable manifold and smooth actions of $\mathbb{R}$ on $\mathcal{M}$, which is given in terms of the standard exponential map
\begin{eqnarray}
\psi : \mathbb{R} &\mapsto &\mathrm{Diff}\left(\mathcal{M}\right)\nonumber\\
t &\mapsto & \psi_{t} = e^{t v}\,, 
\end{eqnarray}

\noindent
and differentiation
\begin{equation}
\psi\mapsto v_{p} =\left. \frac{d\psi_{t}(p)}{dt}\right|_{t = 0}\, , \qquad p\in\mathcal{M}\, .
\end{equation}

\noindent
If the orbits closes up after a fixed time, say $2\pi$, then we actually have a $S^{1}$-action, or equivalently, a $U(1)$-action, on $\mathcal{M}$, namely
\begin{eqnarray}
\psi_{S^1} : S^1 &\mapsto &\mathrm{Diff}\left(\mathcal{M}\right)\nonumber\\
e^{i\theta} &\mapsto & \psi_{\theta} = e^{\theta v}\,.
\end{eqnarray}

\noindent
We will now use the following theorem to give a sufficient condition for $v$ to define a $S^1$ fibration on $\mathcal{M}_{8}$. 

\begin{prop}
\label{thm:actionproperfree}
Let $G$ be a Lie group and $\psi_{G} : G\to\mathrm{Diff}\left(\mathcal{M}\right)$ a free and proper Lie-group action on a differentiable manifold $\mathcal{M}$. Then $\mathcal{M}/G$ admits a unique differentiable structure such that $\mathcal{M}\to\mathcal{M}/G$ is a principal bundle with fibre $G$. In addition, $\mathrm{dim}\left(\mathcal{M}/G\right) = \mathrm{dim}\left(\mathcal{M}\right) - \mathrm{dim}\left( G\right)$.
\end{prop}

\begin{proof}
See theorem 5.119 in \cite{JMLee}.
\end{proof} 

\begin{remark}\label{rem:freeprop}
A free action of $S^1$ on a manifold $\mathcal{M}$ automatically defines an $S^1$-principal bundle structure on $\mathcal{M}$. This follows from proposition \ref{thm:actionproperfree}, since an action of a compact group is always proper.
\end{remark}

\noindent
The following theorem and its proof are analog to theorem \ref{thm:Spin(7)fromG2}. It essentially says that a positive three-form on the base manifold of an $S^1$-bundle can be pulled back to the total space, and gives rise to an admissible four-form form there.
 
\begin{thm}\label{prop:basetoM8}
Let $\mathcal{M}_{8}$ be an eight-dimensional   manifold equipped with 
a free action of $S^{1}$, whose infinitesimal  generator we denote by $v$.

If $\mathcal{M}_{8}/S^1$ is endowed with a topological $G_2$-structure, then there exists a topological $Spin(7)$-structure on $\mathcal{M}_{8}$, such that
the projection $\pi\colon  \mathcal{M}_{8} \to \mathcal{M}_{8}/S^1$ is a Riemannian submersion\footnote{This means that at every point $p$,  $(d_p\pi)|_{{H_7}_p}\colon {H_7}_p\to T_{\pi(p)}(\mathcal{M}_{8}/S^1)$ is an isometry.} w.r.t. the metrics induced by the topological structures.
\end{thm}
\begin{proof}
Let $\psi_3$ be a positive three-form on the seven-dimensional manifold $\mathcal{M}_{8}/S^1$. 
Choose a codimension one distribution $H$ on $\mathcal{M}_{8}$ transverse to $v$. Endow $\mathcal{M}_{8}$ with the unique Riemannian metric $g$ that makes $\pi$ into a Riemannian submersion (w.r.t. the metric on $\mathcal{M}_{8}/S^1$  induced by $\psi_3$)
and which satisfies $||v||=1$, $v\perp H$. Define $\phi_3:=\pi^*\psi_3$. 
Then

\begin{equation*}
\Omega = v^{\flat}\wedge  \phi_3+ \ast\left(v^{\flat}\wedge  \phi_3\right)\, ,
\end{equation*}

\noindent
is an admissible four-form on $\mathcal{M}_{8}$, whose
associated metric is exactly $g$.
This follows applying lemma \ref{lem:phitoOmega} to the vector space $V:=T_p\mathcal{M}_{8}$ and subspace $H:=H_p$, for all $p\in \mathcal{M}_{8}$, and by remark \ref{rem:innonV} (or, which is the same, it follows from theorem \ref{thm:Spin(7)fromG2}).
\end{proof}

\begin{remark}
The admissible four-form in theorem \ref{prop:basetoM8} can always be chosen to be $S^1$-invariant. Indeed  in the proof of theorem \ref{prop:basetoM8} we can always choose the distribution $H$ to be $S^1$-invariant (take $H$ to be the kernel of a connection one-form for the $S^1$-principal bundle). Then the metric $g$ will be $S^1$-invariant. As $\Omega$ is constructed out of $v$, the pullback form $\phi_3$, and $g$, it follows that $\Omega$ is  $S^1$-invariant.
\end{remark}

The next theorem is  a converse to theorem \ref{prop:basetoM8}.

 \begin{thm}
Let $(\mathcal{M}_{8},g)$ be an 
 oriented Riemannian spin manifold with 
a topological $Spin(7)$-structure (whose associated metric is  $g$), and denote by  $\Omega$ the corresponding admissible four-form. 
 Assume a free action of $S^{1}$ on $\mathcal{M}_{8}$ such that  $\Omega$ is $S^1$-invariant.

Then the base $\mathcal{M}_{8}/S^{1}$ is endowed with a canonical positive three-form $ \psi_3$. 
In particular,  $\mathcal{M}_{8}/S^{1}$ has a canonical $G_{2}$-structure. Further, 
the projection $\pi\colon  \mathcal{M}_{8} \to \mathcal{M}_{8}/S^1$ is a Riemannian submersion w.r.t. the metrics induced by the topological structures.
\end{thm}

 
\begin{proof} Denote by $v$ the infinitesimal  generator of the  $S^1$-action.
  Since the $S^1$ action preserves $\Omega$, it preserve also the metric $g$ associated to $\Omega$. Hence it preserves the dimension seven distribution $H_7$ orthogonal to $v$. As obviously $\mathcal{L}_{v}v=0$,
 it follows that the $S^1$ action preserves $(\iota_v\Omega)|_{H_7}$. 
 Notice that by 
remark \ref{rem:noint},
  at every point $p$, the form $(\iota_v\Omega)|_{{H_7}_p}$ is positive, and the metric on ${{H_7}_p}$ associated to this positive three-form is the restriction of $g$.
 Further, the $S^1$ action preserves the unique\footnote{The construction of $\phi_3$  is exactly the construction carried out in 
 proposition \ref{prop:smooth}.}   three form $\phi_3\in \Omega^3(\mathcal{M}_{8})$ which annihilates $v$ and restricts to $(\iota_v\Omega)|_{H_7}$ on $H_7$.

 Since   $\mathcal{L}_v\phi_3=0$ and   $\iota_v\phi_3=0$, there exists a unique three-form $ \psi_3 $ on $\mathcal{M}_{8}/S^{1}$  such that $\pi^*(\psi_3)=\phi_{3}$. At every $p\in \mathcal{M}_{8}$, the isomorphism $(d_p\pi)|_{{H_7}_p} \colon {H_7}_p\to T_{\pi(p)}(\mathcal{M}_{8}/S^{1})$ identifies $\phi_3|_{{H_7}_p}$ with $ \psi_3|_{\pi(p)} $. Since the former is a positive three-form, the latter also is. Further,
  $(d_p\pi)|_{{H_7}_p}$ identifies the metric associated to $\phi_3|_{{H_7}_p}$ with the metric associated $ \psi_3|_{\pi(p)}$. 
  \end{proof}

\section{General intermediate manifolds and elliptic fibrations}
\label{sec:generalintermediate}


In this section we are going to consider a more general choice of intermediate structure $\mathfrak{S}=\left( g, \phi,\tilde{\phi}\right)$ than in section \ref{sec:nd}. In particular, we are not going to take $\tilde{\phi}$ as the dual of $\phi$. Nevertheless, we will still make a particular choice for $\phi$ and $\tilde{\phi}$. Indeed, 
we will take them\footnote{Here we abuse notation and denote 
the dual of $\tilde{\phi}\in\Gamma\left(\Lambda^3 T\mathcal{M}_{8}\otimes T^{\ast}\mathcal{M}_{8}\right)_{0}$ with respect to the metric, which is a section of $ \left(\Lambda^3 T^{\ast}\mathcal{M}_{8}\otimes T\mathcal{M}_{8}\right)_{0}$, by the same symbol $\tilde{\phi}$.} to be non-degenerate (see definition \ref{def:nondeg}):

\begin{equation}
\phi = \phi_{3}\otimes v\, , \qquad \tilde{\phi} = \tilde{\phi}_{3}\otimes \tilde{v}\,.
\end{equation}

\noindent
Notice that in particular, as it happened in section \ref{sec:nd}, the trace-zero condition reads

\begin{equation}
\iota_{v}\phi_{3} = 0\, , \qquad \iota_{\tilde{v}}\tilde{\phi}_{3} = 0\, .
\end{equation}

\noindent
In order to ease the presentation, let us define from now on

\begin{equation}
\phi^{1}_{3} = \phi_{3}\, , \qquad \phi^{2}_{3}=\tilde{\phi}_{3}\, ,\qquad v_{1} = v\, , \qquad v_{2} = \tilde{v}\, .
\end{equation}

\noindent
Having two non-degenerate intermediate structures on a manifold $\mathcal{M}_{8}$ gives rise to many geometrical interesting situations. This section is intended to give the reader just a first glance at the kind of geometric structures that appear in the presence of two non-degenerate structures. We will prove in a moment a general result about the reduction of the topological structure group of the generalized bundle $T\mathcal{M}_{8}\oplus T^* \mathcal{M}_{8}$ over $\mathcal{M}_{8}$. More importantly, in subsections \ref{sec:regularelliptic} and \ref{sec:singularities} we explain how two non-degenerate structures can be intimately related to regular as well as singular elliptic fibrations.

Let us assume that each of $\left(\phi^{a}_{3}, v_{a}\right)\, , a=1,2\, ,$  defines a $G_{2}$-intermediate structure (see definition \ref{def:ndinter}).
Then by theorem \ref{thm:Spin(7)fromG2} there is a $Spin(7)$ structure on $\mathcal{M}_{8}$ associated to each of the admissible four-forms

\begin{equation}
\Omega^{a} = v^{\flat}_{a}\wedge \phi^{a}_{3} + \ast\left(v^{\flat}_{a}\wedge \phi^{a}_{3}\right)\, ,
\end{equation}

\noindent
which in turn implies the existence of the corresponding $Spin(8)$ spinor $\chi_{a}\in\Gamma\left(S^{\mathbb{R}+}_{8}\right)$ on $\mathcal{M}_{8}$. We arrive thus at the following result. 

\begin{prop}
\label{prop:generalizedSpin7}
Let $\left(\mathcal{M}_{8},\mathfrak{S} = \left(g, \phi^{1}_{3}\otimes v_{1}, \phi^{2}_{3}\otimes v_{2}\right)\right)$ be an intermediate manifold such that $\left( g,\phi^a_{3}\otimes v_a \right), a = 1,2\, ,$ is a $G_{2}$-intermediate structure. Then, the structure group of the generalized bundle $\mathbb{E} = T\mathcal{M}_{8}\oplus T^{\ast}\mathcal{M}_{8}$ admits a reduction from $\mathbb{R}^{\ast}\times Spin(8,8)$ to $Spin(7)\times Spin(7)$, that is, it admits a generalized $Spin(7)$-structure.
\end{prop}

\begin{proof} Thanks to theorem \ref{thm:Spin(7)fromG2}, each of the $G_{2}$-intermediate structures $\left( g,\phi^a_{3}\otimes v_a \right), a = 1,2\, ,$ defines a topological $Spin(7)$-structure on $\mathcal{M}_{8}$. Let us denote by $\chi_{a}\in\Gamma\left(S^{\mathbb{R}+}_{8}\right)$ the corresponding, globally defined, spinors. By theorem 5.1 in reference \cite{2006CMaPh.265..275W} we obtain a generalized $Spin(7)$-structure, uniquely determined by

\begin{enumerate}

\item The orientation of $\mathcal{M}_{8}$.

\item The metric $g$ present in the intermediate structure.

\item The vanishing two-form $B=0$.

\item The vanishing scalar function $f=0$.

\item Two Majorana-Weyl spinors $\chi_{a}\in\Gamma\left(S^{\mathbb{R}+}_{8}\right)\, , a = 1,2\,$ such that

\begin{equation}
\rho = \chi_{1}\otimes\chi_{2}\, ,
\end{equation}

\noindent
is an invariant $Spin(7)\times Spin(7)$ spinor.
\end{enumerate}

\end{proof} 

\begin{remark}
Theorem 5.1 in reference \cite{2006CMaPh.265..275W} states that a generalized $Spin(7)$-structure on an eight-dimensional manifold is equivalent to the following data:

\begin{enumerate}

\item An orientation.

\item A metric $g$.

\item A two-form $B$.

\item A scalar function $f$.

\item Two Majorana-Weyl spinors $\chi_{a}\, , a = 1,2\,$ either of the same or different chirality such that 

\begin{equation}
\rho = e^{-f} e^{B}\wedge\chi_{1}\otimes\chi_{2}\, ,
\end{equation}

\noindent
is an invariant $Spin(7)\times Spin(7)$ spinor.
\end{enumerate}

\noindent
In proposition \ref{prop:generalizedSpin7} we have considered the simplest possibility, namely we have taken $f = B = 0$ and the two spinors of the same chirality. Notice however that, given an intermediate structure $\mathfrak{S} = \left(g, \phi^{1}_{3}\otimes v_{1}, \phi^{2}_{3}\otimes v_{2}\right)$ we have a natural scalar function $f$, given in coordinates by

\begin{equation}
f = \phi^{1}_{3\, abc}\, \phi^{2\, abc}_{3}\, ,
\end{equation}

\noindent
as well as two two-forms $B_{1}$ and $B_{2}$ given by

\begin{equation}
B_{1} =  \iota_{v_{2}}\phi^{1}_{3}\, , \qquad B_{2} =  \iota_{v_{1}}\phi^{2}_{3}\, .
\end{equation}

\noindent
The viability of the previous choices for $B$ and $f$ in the definition of a generalized $Spin(7)$-structure remains yet to be understood, but it points out to a perhaps deep connection between intermediate structures and generalized $Spin(7)$-structures.
\end{remark}


\subsection{Regular elliptic fibrations}
\label{sec:regularelliptic}


\noindent
The existence of an intermediate structure $\mathfrak{S}=\left( g, \phi,\tilde{\phi}\right)$ not only implies a topological reduction of the structure group of the generalized bundle, but it may also imply a reduction on the spin bundle of $\mathcal{M}_{8}$.  In particular, the existence of the two Majorana-Weyl spinors $\chi_{a}$ on $\mathcal{M}_{8}$  implies a reduction of the structure group of the spin bundle $S^{\mathbb{R}+}_{8}\to\mathcal{M}_{8}$ to different groups depending on whether they are linearly independent at every point or not. If they are, and they have the same chirality, then the structure group is reduced further from Spin(7) to $SU(4)$, while if they become dependent at some points in $\mathcal{M}_{8}$ then there exists no further global reduction. Notice that the existence of each $\chi_{a}$ is ensured by the existence of the corresponding $Spin(7)$ structure with associated admissible four-form $\Omega^{a}$. If $\chi_{1}$ and $\chi_{2}$ are of different chirality, then the structure group is reduced to $G_{2}$ instead of $SU(4)$.

The combined manifold and structure $\left(\mathcal{M}_{8}, \mathfrak{S}=\left( g, \phi,\tilde{\phi}\right)\right)$ is thus equipped with two globally defined vector fields, $v_{1}$ and $v_{2}$. Each of them defines a codimension-one distribution which, if integrable, gives rise to a family of $G_{2}$-structure manifolds, as we have seen in section \ref{sec:G2manifolds}. Being $\mathcal{M}_{8}$ an eight-dimensional manifold which can have $SU(4)$ structure under a mild assumption, the question is if the intermediate structure $\mathfrak{S}$ may be used to define also an elliptic fibration in $\mathcal{M}_{8}$. Hence, the goal of this subsection is to define an elliptic fibration in $\mathcal{M}_{8}$ from an intermediate structure. For phenomenologically interesting F-theory applications, the elliptic fibration must be singular, meaning that the fibre at some points of the base space is not smooth torus. Here however we are going to consider only regular elliptic fibrations, commenting about how singularities can be implemented in this set up just at the very end, while leaving a complete analysis for a future publication \cite{holonomyMCM}. By regular elliptic fibration we mean a principal torus bundle with total space $\mathcal{M}_{8}$ and six-dimensional base space that we will denote by $\mathcal{B}$. Given the structure $\mathfrak{S}$ that we have defined in our manifold $\mathcal{M}_{8}$, namely two three-forms $\phi^{a}_{3}$ and two vector fields $v_{a}$, it is natural to try to define the elliptic fibration by means of the two-dimensional distribution generated at each point by $v_{1}$ and $v_{2}$. In order to see if this is a sensible way to proceed, it is natural to ask if every elliptic fibration is equipped with two globally defined vector fields defining an integrable distribution. The following proposition answers this question in an affirmative way.

\begin{prop}
\label{prop:T2tov}
Let $\left(T^2, \mathcal{M}, \mathcal{B}\right)$ be a principal torus bundle, with total space $\mathcal{M}$, base $\mathcal{B}$ and fibre $T^2$. Then $X$ is equipped with two globally defined vector fields $v_{1}$ and $v_{2}$ that define a completely integrable two dimensional distribution. In fact, $\left[ v_{1}, v_{2}\right] = 0$.
\end{prop}

\begin{proof}Since $\left(T^2, \mathcal{M}, \mathcal{B}\right)$ is a principal bundle, there is a smooth free action of $T^2 = \mathbb{R}/\mathbb{Z}\times\mathbb{R}/\mathbb{Z}$ on $X$, namely 

\begin{eqnarray}
\psi : \mathbb{R}/\mathbb{Z}\times\mathbb{R}/\mathbb{Z} &\mapsto &\mathrm{Diff}\left(\mathcal{M}\right)\nonumber\\
([t],[s]) &\mapsto & \psi_{([t],[s])}\, ,
\end{eqnarray}

\noindent
Therefore, for every $p\in\mathcal{M}_{8}$, we can define two curves $\gamma^1_{p}\colon S^{1}\to\mathcal{M}_{8}$ and $\gamma^2_{p}\colon S^{1}\to\mathcal{M}_{8}$ as follows

\begin{equation}
\gamma^1_{p}([t]) =  \psi_{([t],[0])}( p) , \qquad \gamma^2_{p}([s]) =  \psi_{([0],[s])}( p).
\end{equation}  

\noindent
They satisfy $\gamma^1_{p}([0]) =\gamma^2_{p}([0]) = p$. At every point $p\in\mathcal{M}_{8}$ we define then

\begin{equation}
\left. v_{1}\right|_{p} = \frac{d}{dt}\gamma^1_{p}([0])\, , \qquad  \left. v_{2}\right|_{p} = \frac{d}{ds}\gamma^2_{p}([0])\, .
\end{equation}  

\noindent
$v_{1}$ and $v_{2}$ are the infinitesimal generators of the torus action. Since the infinitesimal action is a Lie-algebra homomorphism, we must have $[v_{1},v_{2}] = 0$.
\end{proof}

Therefore, it is reasonable to try to define an elliptic fibration by using two vector fields. We can even go further and completely characterize an elliptic fibration in terms of two globally defined vector fields.

\begin{prop}
\label{prop:v1v2T2}
$\mathcal{M}$ admits an elliptic fibration if and only if it is equipped with two globally defined, linearly independent vector fields $v_{1}$ and $v_{2}$, such that:
\begin{enumerate}
\item $\left[ v_{1}, v_{2}\right] = 0$,
\item all the leaves of the foliation integrating the rank two distribution $\mathcal{H}:=span\{v_1,v_2\}$ are compact,
\item the leaf space $\mathcal{M}/\mathcal{H}$ of the foliation is a smooth manifold,
\item $\Lambda$ is a trivial bundle, where $\Lambda$ is the bundle   of isotropy groups of the $\mathbb{R}^2$-action on $\mathcal{M}$ generated by $v_1$,$v_2$.
\end{enumerate}
\begin{remark}
a) The bundle $\Lambda$ appearing in assumption 4 is a bundle of rank two lattices over $\mathcal{M}/\mathcal{H}$, contained in the trivial vector bundle $\mathbb{R}^2 \times \mathcal{M}/\mathcal{H}$. (In particular each fiber of $\Lambda$, as a group, is isomorphic to $\mathbb{Z}^2$.)
If the leaf space $\mathcal{M}/\mathcal{H}$ happens to be simply connected, condition $4.$ is automatically satisfied.

b) The flows of $v_1$ and $v_2$ as above are not periodic in general.
Even when they are periodic of minimal period one, the principal two-torus action on $\mathcal{M}$ given by the above proposition is not the product of the $S^1$-action generated by $v_1$ with the  $S^1$-action generated by $v_2$. To see this, consider the case where $\mathcal{M}=S^1\times S^1$,
and $v_1=\frac{\partial}{\partial \theta_2}$, $v_2=2\frac{\partial}{\partial \theta_1}+\frac{\partial}{\partial \theta_2}$: the induced two-torus action is not free, for the time one flow of $\frac{1}{2}(v_1+v_2)$ is the identity.
\end{remark}
\end{prop}
 
\begin{proof} 
Given propositions \ref{prop:T2tov}, we just need to assume the existence of vector fields $v_{1}$ and $v_{2}$ as above and show that $\mathcal{M}$ admits an elliptic fibration. Since the leaves of $\mathcal{H}$ are compact, the vector fields $v_1$ and $v_2$ are complete, hence by assumption  $1$ they generate a $\mathbb{R}^2$-action on $\mathcal{M}$ whose orbits are exactly the leaves of $\mathcal{H}$. At every $p\in  \mathcal{M}$, the isotropy group of this action is a rank two lattice in
$\mathbb{R}^2$, for the orbits of the  $\mathbb{R}^2$ are compact. Further, the isotropy group at $p$ is equal to the isotropy group at any other point in the orbit through $p$, since $\mathbb{R}^2$ is an abelian group. We denote this isotropy group by $\Lambda_{\pi(p)}$, where   $\pi\colon \mathcal{M}\to \mathcal{M}/\mathcal{H}$ is the projection to the leaf space. 

At  every $u\in  \mathcal{M}/\mathcal{H}$ we can choose a basis (over $\mathbb{Z}$)  of the lattice $\Lambda_u$, which we denote by $\{(a_1(u),a_2(u)),\; (b_1(u),b_2(u))\}$.
The bundle $\Lambda$ over $\mathcal{M}/\mathcal{H}$ is trivial by assumption $4.$
Hence   this basis can be chosen to depend smoothly on $u\in \mathcal{M}/\mathcal{H}$, that is, so that  $a_1,a_2,b_1,b_2\in C^{\infty}(\mathcal{M}/\mathcal{H})$.
Define the following two vector fields on $\mathcal{M}$:
$$A:=\pi^*(a_1)v_1+\pi^*(a_2)v_2,\;\;\;\;\;\;\;\;\;\;\;\;\;B:=\pi^*(b_1)v_1+\pi^*(b_2)v_2.$$
These two vector fields commute. They generate an action of the torus $S^1\times S^1$ on $\mathcal{M}$, since the integral curves of $A$ and $B$ are periodic of period one. This torus action is free, since $A$ and $B$ are constructed out of a basis of $\Lambda$, and its orbits are exactly the leaves of $\mathcal{H}$.
Hence it defines an elliptic fibration by proposition \ref{thm:actionproperfree}.
\end{proof} 

\noindent
From propositions \ref{prop:T2tov} we conclude that using the two vector fields present in an intermediate manifold to define an elliptic fibration in terms of an integrable distribution is the sensible way to proceed. Therefore, let $\left(\mathcal{M}_{8}, \mathfrak{S}=\left( g, \phi^{1}_{3}\otimes v_{1}, \phi^{2}_{3}\otimes v_{2}\right)\right)$ be an intermediate manifold. We define the following two-dimensional distribution 

\begin{equation}
\label{eq:2distribution}
H_{2} = \left\{ H_{2\, p}\, , \quad p\in\mathcal{M}_{8}\right\}\, ,
\end{equation}

\noindent
as follows

\begin{equation}
H_{2\, p} = \mathrm{span}\left( v_{1}|_{p},v_{2}|_{p}\right)\, .
\end{equation}

\noindent
By Frobenius' theorem, $H_{2}$ will be completely integrable if and only if

\begin{equation}
\left[ w_{1}, w_{2}\right]\in \Gamma\left(H_{2}\right)\, , \quad\forall\, w_{1}, w_{2} \in \Gamma\left( H_{2}\right)\, .
\end{equation}

\noindent

\noindent
Now, having an elliptic fibration is equivalent to having $v_{1}$ and $v_{2}$ commuting and generating the infinitesimal, free, action of a torus on $\mathcal{M}_{8}$, and therefore the problem is fully characterized.

\begin{remark}
We may wonder how many different two-dimensional foliations we have if we require the distribution (\ref{eq:2distribution}) to be integrable but we do not necessarily require it to correspond to an elliptic fibration. In that case we obtain that $\mathcal{M}_{8}$ is foliated by parallelizable, and therefore orientable, two-dimensional manifolds. As a consequence, if the foliation is by compact leaves, then they must be oriented surfaces with zero Euler characteristic, that is, they must be elliptic surfaces. 
\end{remark}

\noindent
At this point we have completely characterized regular elliptic fibrations in terms of two vector fields on $\mathcal{M}_{8}$. However, we want to go further, and we want to study the possibility of having singular fibres at a set of points $S_{I}\subset\mathcal{B}$ on the base space $\mathcal{B}$. For F-theory applications, the space $S_{I}$ has to satisfy some extra-requeriments, in particular, in the simplest case $S_{I}$ must be of complex codimension one inside $\mathcal{B}$. Let $A$ and $B$ be the canonical basis of the homology group $H_{1}\left( T^{2}\right)$ of the torus. Then, if at every point in  $b\in S_{I}$ there is a combination of $A$ and $B$ such that the cycle

\begin{equation}
C = p A + q B\, , \qquad p, q\in\mathbb{Z}\, , \qquad \mathrm{m.c.d}\left(p_,q\right) = 1\, ,
\end{equation} 

\noindent
vanishes, one can conclude then that, from the physics point of view, there is a $\left( p, q\right)$-seven-brane that extendes over the four non-compact dimensions and wraps the four compact dimensions of $S_{I}$. We will see in section \ref{sec:singularities} how to implement this kind of singularities just in terms of two vector fields, by using the map that we are going to construct now between a subset of the vector fields of $\mathcal{M}_{8}$ and $H_{1}\left( T^{2}\right)$.

Let us assume then that we have a manifold $\mathcal{M}_{8}$ equipped with an intermediate structure $\mathfrak{S}=\left( g, \phi^{1}_{3}\otimes v_{1}, \phi^{2}_{3}\otimes v_{2}\right)$ such that $v_{1}$ and $v_{2}$ are the infinitesimal generators of a free torus action. We thus know that there is a free torus action on $\mathcal{M}_{8}$

\begin{equation}
\psi\colon T^2 = \mathbb{R}/\mathbb{Z}\times\mathbb{R}/\mathbb{Z}\to \mathrm{Diff}\left(\mathcal{M}_{8}\right)\, ,
\end{equation}

\noindent
with the corresponding infinitesimal action given by the Lie-algebra homomorphism

\begin{eqnarray}
v_{-}\colon \mathbb{R}\times\mathbb{R} &\to & \mathfrak{X}\left(\mathcal{M}_{8}\right)\, , \\
(t,s) &\mapsto & v_{t,s}\, \nonumber  ,
\end{eqnarray}

\noindent
where 

\begin{equation}
\left. v_{t,s}\right|_{p} = \left.\frac{d}{d\alpha}\right|_{\alpha=0}\, e^{\alpha(t,s)}( p) , \qquad \forall p\in\mathcal{M}_{8}\, .
\end{equation}

\noindent
Since $\mathbb{R}\times\mathbb{R}$ is equipped with the trivial Lie-bracket, the image of $v_{-}$ is given by linear combinations of $v_{1}$ and $v_{2}$ over $\mathbb{R}$ as follows

\begin{equation}
v_{t,s} =  t\, v_{1} + s\, v_{2}\, ,\qquad t, s\in\mathbb{R}\, ,
\end{equation}

\noindent
where we have identified $v_{1} = v_{(1,0)}$ and $v_{2} = v_{(0,1)}$. For every $p\in\mathcal{M}_{8}$ there is a unique curve $\gamma^{p}_{t,s}\colon \mathbb{R}\to\mathcal{M}_{8}$ such that 

\begin{equation}
\gamma^{p}_{t,s}(0) = p\, , \qquad \left.\frac{d}{d\tau}\gamma^{p}_{t,s}(\tau)\right. = \left. v_{t,s}\right|_{\gamma^{p}_{t,s}(\tau)}\, . 
\end{equation}

\noindent
Let us denote by $\pi\colon\mathcal{M}_{8}\to\mathcal{B} = \mathcal{M}_{8}/T^2$ the projection of the torus bundle. Then, from the standard properties of principal bundles we have

\begin{equation}
\gamma^{p}_{t,s}(\tau)\in\pi^{-1}(p)\simeq T^2\, , \quad \forall p\in\mathcal{M}_{8}\, ,\quad \forall \tau\in\mathbb{R}\, .
\end{equation}

\noindent
By assumption, the flow of $v_{1}$ and $v_{2}$ is a closed curve in $\pi^{-1}(p)\simeq T^2$ for every $p\in\mathcal{M}_{8}$. However, not every vector field $v_{t,s}$ will give rise to a closed curve in the corresponding torus. The condition for the flow of $v_{t,s}$ to be a closed curve is given by

\begin{equation}
\frac{t}{s}\in\mathbb{Q}\, .
\end{equation}

\noindent
In particular, if we want the image of the closed curve $\gamma^{p}_{t,s}$ to be covered just one time when $\tau$ goes from zero to one then\footnote{We thank Raffaele Savelli for a clarification about this point.}

\begin{equation}
\mathrm{m.c.d}\left(s,t\right) =1\, ,
\end{equation} 

\noindent
that is, $s$ and $t$ must be coprime. Let us use the notation $\mathfrak{X}^{c}_{s,t}\left(\mathcal{M}_{8}\right) = \left\{ v_{s,t}\, , \, \, |\,\, \mathrm{m.c.d}\left(s,t\right) =1\right\}$. Then, for every $p\in\mathcal{M}_{8}$, there is a well-defined map $\delta^{p}$ from the set of vector fields $v_{s,t}$ such that $\mathrm{m.c.d}\left(s,t\right) =1$ to the first homology group $H_{1}\left( T^2\right)$ of the torus, given by
\begin{eqnarray}
\label{eq:deltamap}
\delta^{p}\colon \mathfrak{X}^{c}_{s,t}\left(\mathcal{M}_{8}\right) &\to & H_{1}\left( T^2\right)\, ,\\
v_{s,t} &\mapsto & \left[\gamma^{p}_{t,s}\right]\, \nonumber ,
\end{eqnarray}

\noindent
Let us take as a canonical basis $v_{1,0}$ and $v_{0,1}$, and therefore under $\delta^{p}$ we have that $v_{1,0} \mapsto A=\left[\gamma^{p}_{1,0}\right]$ and $v_{0,1} \mapsto B=\left[\gamma^{p}_{0,1}\right]$, where $A$ and $B$ are the standard generators of the homology group $H_{1}\left( T^2\right) = \mathbb{Z}\times\mathbb{Z}$. 


\subsection{Possibility of singular points}	
\label{sec:singularities}

	
So far we have considered manifolds where there is at least one globally defined, no-where vanishing vector field.	As a consequence, if $\mathcal{M}_{8}$ is closed, we would conclude that the Euler characteristic must be zero: $\chi\left(\mathcal{M}_{8}\right) = 0$. 
This seems to be a too restrictive condition. However, we should not be concerned, for two reasons. On one hand, we do not necessarily require $\mathcal{M}_{8}$ to be compact, meaning that $\mathcal{M}_{8}$ does not have to be necessarily understood as a compactification space. In fact,  if we consider $\mathcal{M}_{8}$ to be compact it would only be in order to make contact with eight-dimensional compact internal spaces in F-theory. On the other hand, as we have mentioned before, even if we take $\mathcal{M}_{8}$ to be compact, phenomenologically interesting F-theory applications are based on elliptic fibrations where the fibre becomes singular at given set of points on the base space $\mathcal{B}$, including the possibility of having a singular space $\mathcal{M}_{8}$. In this case we would not require the vector fields $v_{a}$ to be globally non-vanishing nor $\mathcal{M}_{8}$ to be a smooth manifold, and therefore again the Euler number does not have to vanish. All the previous results in this paper have been derived assuming that the vector fields $v_{a}$ are globally defined and nowhere vanishing.  Hence, in order to include the possibility of having a singular elliptic fibration, we will study now the situation where we admit linear combinations of the vector fields $v_{1}$ and $v_{2}$ to have zeros on $\mathcal{M}_{8}$ (that is,
$v_1$ and $v_2$ are linearly dependent at a set of points of $\mathcal{M}_{8}$). More precisely, let us consider the case in which  the vector field

\begin{equation}
v_{s,t} = s\, v_{1} +  t\, v_{2}\, ,
\end{equation} 

\noindent 
is zero at $p\in S_{I}\subset\mathcal{B}$ for the particular combination of $v_{1}$ and $v_{2}$ given by  $\left(s,t\right)$. Therefore, from equation (\ref{eq:deltamap}) we deduce that the maps

\begin{eqnarray}
\label{eq:deltamap2}
\delta^{p}\colon \mathfrak{X}^{c}_{s,t}\left(\mathcal{M}_{8}\right) &\to & H^{1}\left( T^2\right)\, ,\\
v_{s,t} &\mapsto & \left[\gamma^{p}_{t,s}\right]\nonumber \,
\end{eqnarray}

\noindent
are singular at the points $v_{s,t}$. Hence, we conclude that at every point $p\in S_{I}$ the torus fibration is singular since the cycle $s\, A + t\, B$ of the torus $\pi^{-1}(p)$ collapses. We are thus able to see where the fibre becomes singular  just from the analysis of the zeros of the vector field $v_{t,s}$. If $S_{I}$ satisfy the corresponding requeriments, then it can be interpreted from the physics point of view as a set of $D7$-branes extending on the four non-compact dimensions and four compact dimensions on the base space $\mathcal{B}$.


\section{Conclusions}
\label{sec:conclusions}


In this paper we have studied ${\cal N}=1$ M-theory compactifications down to four dimensions in terms of  an eight-dimensional manifold $\mathcal{M}_{8}$ endowed with an intermediate structure $\mathfrak{S}$, whose presence is motivated by the exceptionally generalized geometric description of  such compactifications. We have restricted our attention to off-shell supersymmetry, namely to the topological implications of $\mathfrak{S}$ on $\mathcal{M}_{8}$. We have found that using $\mathfrak{S}$ together with some mild assumptions, it is possible to embed in $\mathcal{M}_{8}$ a family of $G_{2}$-structure seven-dimensional manifolds as the leaves of a codimension-one foliation. At the same time, it is possible to prove that $\mathcal{M}_{8}$ is equipped with a topological $Spin(7)$-structure. This, if explored further, may give a relation between seven-dimensional manifolds that are consistent internal spaces in M-theory,  and eight-dimensional manifolds in F-theory compactifications, perhaps pointing out to some kind of duality between M-theory compactified, loosely speaking, on the leaves of the foliation, and F-theory compactified on $\mathcal{M}_{8}$. With a different set of assumptions, it is possible to use $\mathfrak{S}$ to naturally define an $S^{1}$ principal-bundle structure on $\mathcal{M}_{8}$, which turns out to be equipped with a topological $Spin(7)$-structure if the base space is equipped with a topological $G_{2}$-structure.  One can go further and, considering a more general intermediate structure $\mathfrak{S}$, define in $\mathcal{M}_{8}$ an elliptic fibration. In that case, if the base space of the elliptic fibration has a topological $SU(3)$-structure, then $\mathcal{M}_{8}$ has a topological $Spin(7)$-structure\footnote{The proof of this statement will be presented elsewhere, as it involves topological $SU(3)$-structures, which were not discussed in this letter in order not to obscure the presentation.}. In addition, we show that the elliptic fibration can be completely characterized through the pair of vector fields $v_{1}$ and $v_{2}$ that are present in $\mathfrak{S}$, and that the possible singularities on the elliptic fibration correspond to zeros of a particular combination of $v_{1}$ and $v_{2}$. 

We have studied here some of the implications drawn by the existence of an intermediate structure $\mathfrak{S}$ and, more generically, we have explored the relation between eight-dimensional $Spin(7)$ and seven-dimensional $G_{2}$-manifolds. This might be of great physical interest, in particular for the study of dualities among String/M/F-Theory compactifications.  For that, further work needs to be done, mainly to understand how the differential conditions on $\mathfrak{S}$ required by on-shell supersymmetry imply reductions of the holonomy group of $\mathcal{M}_{8}$ or its preferred submanifolds.

\acknowledgments

We would like to thank  Marcos Alexandrino, Vicente Cort\'es, Thomas Grimm, Nigel Hitchin, Dominic Joyce, Ruben Minasian, David Morrison, Raffaele Savelli, Dirk T\"oben, Daniel Waldram and Frederik Witt for useful discussions. This work was supported in part by the ERC Starting Grant 259133 -- ObservableString, and by projects  
MTM2011-22612 and ICMAT Severo Ochoa  SEV-2011-0087 (Spain) and
Pesquisador Visitante Especial grant  88881.030367/2013-01 (CAPES/Brazil).

\appendix


\section{$G_{2}$-manifolds}
\label{sec:G2appendix}


This section is dedicated to summarize the main results concerning seven-dimensional manifolds of $G_{2}$-structure and holonomy contained in $G_{2}$. We will closely follow \cite{Joyce2007,2012arXiv1206.3170C}. $G^{\mathbb{C}}_{2}$ is a simply-connected, semisimple complex fourteen-dimensional Lie group\footnote{For more details see \cite{Adams,2009arXiv0902.0431Y}.}. It has two real forms, namely

\begin{itemize}

\item The real compact form $G_{2}\subset SO(7)$.

\item The real non-compact form $G^{\ast}_{2}\subset SO(4,3)$.

\end{itemize}

\noindent
One way to characterize the real forms of $G_{2}$ is by the isotropy groups of the three-forms $\phi_{0}$ (\ref{eq:phi0}) and $\phi_{1}$ (\ref{eq:phi1}). More precisely, the isotropy group of $\phi_{0}$ is the real compact form  $G_{2}$ and the isotropy group of $\phi_{1}$ is the real non-compact form $G^{\ast}_{2}$. 

\noindent
We recall the definition of $\phi_0$ (definition \ref{def:phi0}):
it is the three-form on $\mathbb{R}^{7}$ given by 
\begin{equation}
\phi_{0} = dx_{123} + dx_{145} + dx_{167} + dx_{246} - dx_{257} - dx_{347} - dx_{356}  \, ,
\end{equation}
where $dx_{ijl}$ stands for  $dx_{i}\wedge dx_{j}\wedge  dx_{l}$.

\bigskip

\noindent
Now let $\mathcal{M}$ be an oriented seven-dimensional manifold. 

\begin{definition}
For each $p\in\mathcal{M}$ we define

\begin{equation}
\mathcal{P}^{3}_{p}\mathcal{M} = \left\{ \phi\in \Lambda^{3} T^{\ast}_{p}\mathcal{M} \quad |\quad \phi = f^{\ast}_{p}\phi_{0} \right\}\, ,
\end{equation}

\noindent
where $f_{p}: T_{p}\mathcal{M}\to\mathbb{R}^{7}$ is an oriented\footnote{By this we mean: orientation preserving.} isomorphism.

\end{definition}

\noindent
 $\mathcal{P}^{3}_{p}\mathcal{M}$ is isomorphic to the quotient $GL_{+}\left(7,\mathbb{R}\right)/G_{2}$ which has dimension $49-14=35$. Since the dimension of $\Lambda^{3} T^{\ast}_{p}$ is also 35, we see that $\mathcal{P}^{3}_{p}\mathcal{M}$ is an open subset of $\Lambda^{3} T^{\ast}_{p}$ and therefore a manifold\footnote{Notice that $\mathcal{P}^{3}_{p}\mathcal{M}$ is not a vector space.}. See   proposition \ref{prop:35} for more details.

\begin{definition}
We define $\mathcal{P}^{3}\mathcal{M}\xrightarrow{\pi}\mathcal{M}$ to be the bundle over $\mathcal{M}$ whose fibre at $p\in\mathcal{M}$ is given by $\pi^{-1}(p) = \mathcal{P}^{3}_{p}\mathcal{M}$.
\end{definition}

\noindent
Notice that $\mathcal{P}^{3}\mathcal{M}$ is an open subbundle of $\Lambda^{3} T^{\ast}\mathcal{M}$ with fibre  $GL_{+}\left(7,\mathbb{R}\right)/G_{2}$.

\begin{definition}
A three-form $\phi\in \Omega^3(\mathcal{M})$ is said to be \emph{positive} if $\phi_{p}\in  \mathcal{P}^{3}_{p}\mathcal{M}$ for all $p\in\mathcal{M}$.
\end{definition}

\noindent
Let us denote by $F_{+}\mathcal{M}$ the oriented frame bundle of $\mathcal{M}$, that is, the bundle over $\mathcal{M}$ whose fibre $F_{+}\mathcal{M}_p$ over $p\in\mathcal{M}$ is the set of ordered bases of $T_{p}\mathcal{M}$ inducing the given orientation:
\beq
F_{+}\mathcal{M}_p = \left\{ e_{p} = \left(e_{1},\hdots, e_{7}\right)\, ,\quad | \quad e_p \quad\mathrm{ordered\,\,oriented\,\, basis\,\, of\,\, } T_{p}\mathcal{M}\right\} \ . \nonumber
 \eeq
Notice that such an ordered oriented basis  $e_p$ can be identified with an orientated isomorphism $T_{p}\mathcal{M}\overset{\cong}{\to}\mathbb{R}^{7}$, which sends the $j$-th basis vector of $e_p$ to the $j$-th vector of the standard basis of $
\mathbb{R}^{7}$, for all $j$.
 The Lie group $GL_+\left(7,\mathbb{R}\right)$ acts freely and transitively on each fibre as follows

\begin{equation}
A\cdot e_{p} = \left(p,\sum_{j=1}^{k} A_{ij} e_{j}\right)\, , \quad A\in GL_+\left(7,\mathbb{R}\right)\, , \quad p\in\mathcal{M}\, .
\end{equation}

\noindent
Hence
$F_{+}(\mathcal{M})$ is a principal bundle with typical fibre $GL_+\left(7,\mathbb{R}\right)$. 
Suppose now that $\mathcal{M}$ is equipped with a positive three-form $\phi$. Consider
\begin{equation}\label{eq:qm}
Q\mathcal{M}_p:=\{f_p: T_{p}\mathcal{M}\overset{\cong}{\to}\mathbb{R}^{7} \text{ oriented isomorphism such that } \phi_{p} = f^{\ast}_{p}\phi_{0}\}.
\end{equation}
Since $\phi$ is positive, this is non-empty
at every $p\in  \mathcal{M}$. By identifying an isomorphism $ T_{p}\mathcal{M}\to \mathbb{R}^{7}$ with an ordered basis of $T_{p}\mathcal{M}$ as above,
we can regard $Q \mathcal{M}_p$ as a subset of $F_{+}\mathcal{M}_p$.
 $Q\mathcal{M}$ is a  principal $G_{2}$-subbundle   of $F_+\mathcal{M}$, i.e. a  $G_{2}$-structure compatible with the orientation of $\mathcal{M}$, since by construction $\phi_{0}$ is invariant under the natural action of $G_{2}\subset GL_{+}\left(7,\mathbb{R}\right)$ on $\mathbb{R}^{7}$.   Conversely, given a $G_{2}$-subbundle $Q\mathcal{M}$ of $F_+\mathcal{M}$ we can define a positive three-form on $\mathcal{M}$ using equation \eqref{eq:qm}. Therefore, we have found a one-to-one correspondence between positive three-forms $\phi$ on $\mathcal{M}$ and $G_{2}$-structures on $\mathcal{M}$ compatible with the orientation. 

A Riemannian metric on an orientable manifold $\mathcal{M}$ implies a reduction of the structure group of the frame bundle from $GL_{+}\left(7,\mathbb{R}\right)$ to $SO(7)$. Since $G_{2}\subset SO(7)$, there is an associated metric $g$ to the $G_{2}$ structure on $\mathcal{M}$. We will call a positive form $\phi$ on $\mathcal{M}$ together with its associated metric $g$ a $G_{2}$-structure on $\mathcal{M}$, since, although $(\phi,g)$ is not a $G_{2}$-structure it  uniquely defines one. Let $\nabla$ be the Levi-Civita connection associated to $g$. We call $\nabla\phi$ the torsion of $(\phi,g)$ and we say that $(\phi,g)$ is torsion free if $\nabla\phi = 0$. The following proposition holds.

\begin{prop}
\label{prop:G2holonomyequiv}
Let $\mathcal{M}$ be an oriented seven-dimensional manifold and $(\phi,g)$ a $G_{2}$-structure on $\mathcal{M}$. Then the following are equivalent

\begin{itemize}

\item $(\phi,g)$ is torsion-free.

\item $\mathrm{Hol}\left( g\right)\subseteq G_{2}$ and $\phi$ is the induced three-form.

\item $\nabla\phi =0$ on $\mathcal{M}$, where $\nabla$ is the Levi-Civita connection of $g$.

\item $d\phi = \delta \phi = 0$ on $\mathcal{M}$.


\end{itemize}

\end{prop}

\begin{proof}See lemma 11.5 in \cite{Salamon}.\end{proof} 

\noindent
Here $\delta$ stands for the codifferential, defined as $\delta = \ast d \ast$. Notice that none of the above conditions is linear on $\phi$, since the metric $g$ depends non-linearly on it. This implies that the operators $\nabla$ and $d\ast$ depend on $g$, which in turn depends on $\phi$ and thus the equations $\nabla\phi=0$ and $d\ast\phi=0$ should not be considered as linear in $\phi$.

We show now how torsion free $G_2$ structures arise from Calabi-Yau manifolds of complex dimension two and three respectively.

\begin{prop}
\label{prop:G2R3}
Suppose $\left(Y,g_{Y}\right)$ is a Riemannian four-dimensional manifold with holonomy $SU(2)$. Then $Y$ admits a complex structure form $J$, a K\"ahler form $\omega$ and a holomorphic volume form $\mathcal{V}$ such that $d\omega = d\mathcal{V} = 0$. Let $\mathbb{R}^{3}$ have coordinates $(x_{1},x_{2},x_{3})$ and euclidean metric $h = dx^{2}_{1} + dx^{2}_{2} + dx^{2}_{3}$. Define a metric $g$ and a three-form $\phi$ on $\mathbb{R}^{3}\times Y$ by $g=h\times g_{Y}$ and
\begin{equation}
\phi = dx_{1}\wedge dx_{2}\wedge dx_{3} + dx_{1}\wedge\omega + dx_{2}\wedge {\rm Re} \mathcal{V} - dx_{3}\wedge
{\rm Im}  \mathcal{V}\, .
\end{equation}

\noindent
Then $(\phi, g)$ is a torsion free $G_2$ structure on $\mathbb{R}^{3}\times Y$ and 
\begin{equation}
\ast \phi = \frac{1}{2} \omega\wedge\omega + dx_{2}\wedge dx_{3}\wedge \omega - dx_{1}\wedge dx_{3} \wedge {\rm Re}\mathcal{V} - dx_{1}\wedge dx_{2}\wedge {\rm Im} \mathcal{V}\, .
\end{equation}
\end{prop}

\begin{proof}See proposition 11.1.1 in reference \cite{Joyce2007}.\end{proof} 

\noindent
It possible to substitute $\mathbb{R}^3$ by $T^3$ and obtain a similar result.

\begin{prop}
\label{prop:G2R}
Suppose $\left(Y,g_{Y}\right)$ is a Riemannian six-dimensional manifold with holonomy $SU(3)$. Then $Y$ admits a complex structure form $J$, a K\"ahler form $\omega$ and a holomorphic volume form $\mathcal{V}$ such that $d\omega = d\mathcal{V} = 0$. Let $\mathbb{R}$ have coordinate $x$. Define a metric $g$ and a three-form $\phi$ on $\mathbb{R}\times Y$ by $g=dx^2\times g_{Y}$ and

\begin{equation}
\phi =  dx\wedge\omega + {\rm Re} \mathcal{V} \, .
\end{equation}

\noindent
Then $(\phi, g)$ is a torsion free $G_2$ structure on $\mathbb{R}\times Y$ and 

\begin{equation}
\ast \phi = \frac{1}{2} \omega\wedge\omega + dx \wedge {\rm Im} \mathcal{V}\, .
\end{equation}
\end{prop}

\begin{proof}See proposition 11.1.1 in reference \cite{Joyce2007}.\end{proof} 

\noindent
It possible to substitute $\mathbb{R}$ by $S^{1}$ and obtain a similar result.


\section{$Spin\left(7\right)$-manifolds}
\label{sec:Spin7appendix}


This section is dedicated to summarize the main results concerning eight-dimensional manifolds with $Spin\left(7\right)$-structure and holonomy contained in $Spin\left(7\right)$. We will closely follow \cite{Joyce2007,2012arXiv1206.3170C}.

\begin{definition}
Let us consider $\mathbb{R}^{8}$ with coordinates $\left(x_{1},\hdots,x_{8}\right)$. We define a four-form $\Omega_{0}$ on $\mathbb{R}^{8}$ by 
\begin{align}\label{eq:fourform}
\Omega_{0} &= dx_{1234} + dx_{1256} + dx_{1278} + dx_{1357} - dx_{1368} - dx_{1458}  - dx_{1467} \nonumber \\
& - dx_{2358} - dx_{2367}- dx_{2457} +  dx_{2468} + dx_{3456} + dx_{3478} + dx_{5678} \, ,
\end{align}

\noindent
where $dx_{ij\hdots l}$ stands for the exterior form $dx_{i}\wedge dx_{l}\wedge \cdots\wedge dx_{l}$. The subgroup of $GL\left(8,\mathbb{R}\right)$ that preserves $\Omega_{0}$ is the Lie group $Spin\left(7\right)$\footnote{$Spin\left(7\right)$ is a compact, connected, simply-connected, semisimple an 21-dimensional Lie group, isomorphic as a Lie group to the double cover of $SO\left(7\right)$.}, which also fixes euclidean metric $g_{0} = dx_{1}^{2} + \cdots + dx_{8}^2$ and the orientation on $\mathbb{R}^{8}$ (that is, $Spin\left(7\right)\subset SO(8)$).  Notice that $\ast \Omega_{0} = \Omega_{0}$ where $\ast$ is the Hodge-dual operator associated to $g_{0}$. 

\end{definition}

\begin{definition}\label{def:admissible}
Let $V$ be an oriented eight-dimensional vector space. A four-form $\Omega\in\Lambda^{4} V^{\ast}$ is said to be \emph{admissible} if there exists an oriented isomorphism $f:\mathbb{R}^{8}\to V$ such that $\Omega_{0} = f^{\ast} \Omega$. 
\end{definition}
\noindent
Notice that an admissible form on $V$ induces an inner product on $V$.

\bigskip

\noindent Let $\mathcal{M}$ be an oriented eight-dimensional manifold.
\begin{definition}
 For each $p\in\mathcal{M}$ we define

\begin{equation}
\mathcal{A}^{4}_{p}\mathcal{M} = \left\{ \Omega \in \Lambda^{4} T^{\ast}_{p}\mathcal{M} \quad |\quad \Omega = f^{\ast}_{p}\Omega_{0} \right\}\, ,
\end{equation}

\noindent
where $f_{p}: T_{p}\mathcal{M}\to\mathbb{R}^{8}$ is an oriented isomorphism.

\end{definition}

\noindent
$\mathcal{A}^{4}_{p}\mathcal{M}$ is isomorphic to the quotient $GL_{+}\left(8,\mathbb{R}\right)/Spin\left(7\right)$ which has dimension $64-21=43$. The dimension of $\Lambda^{4} T^{\ast}_{p}$ is 70, and thus $\mathcal{A}^{3}_{p}\mathcal{M}$ has codimension 27 in $\Lambda^{4} T^{\ast}_{p}$ in contrast to the case of $G_{2}$ considered in appendix \ref{sec:G2appendix}, where $\mathcal{P}^{3}_{p}\mathcal{M}$ is open in $\Lambda^{3} T^{\ast}_{p}$. Notice that $\mathcal{A}^{4}_{p}\mathcal{M}$ is not a vector space.

\begin{definition}
We define $\mathcal{A}^{4}\mathcal{M}\xrightarrow{\pi}\mathcal{M}$ to be the bundle over $\mathcal{M}$ whose fibre at $p\in\mathcal{M}$ is given by $\pi^{-1}(p) = \mathcal{A}^{4}_{p}\mathcal{M}$.
\end{definition}

\begin{definition} \label{admissible}
A four-form $\Omega\in \Lambda^{4} T^{\ast}\mathcal{M}$ is said to be \emph{admissible} if $\Omega_{p}\in  \mathcal{P}^{4}_{p}\mathcal{M}$ for all $p\in\mathcal{M}$.
\end{definition}

\noindent
Let us denote by $F_+\mathcal{M}$ the oriented frame bundle of $\mathcal{M}$. Assume $\mathcal{M}$ is equipped with an admissible four-form $\phi$. Since $\Omega$ is admissible, at every point $p\in\mathcal{M}$ there exists an oriented isomorphism
\begin{equation}
f_{p}: T_{p}\mathcal{M}\to\mathbb{R}^{8}\, ,
\end{equation}

\noindent
such that $\Omega_{p} = f^{\ast}_{p}\Omega_{0}$. 
$\Omega$ can be used to define a principal subbundle $Q\mathcal{M}$ of $F_+\mathcal{M}$ as follows:
\begin{equation}\label{eq:qm8}
Q\mathcal{M}_p:=\{f_p: T_{p}\mathcal{M}\overset{\cong}{\to}\mathbb{R}^{8} \text{ oriented isomorphism such that } \Omega_{p} = f^{\ast}_{p}\Omega_{0}\}.
\end{equation}
The structure group of $Q\mathcal{M}$ is $Spin\left(7\right)$, since the stabilizer of 
$\Omega_{0}$  under the natural action of $GL_+\left(8,\mathbb{R}\right)$ on $\mathbb{R}^{8}$ is $Spin\left(7\right)$.
Conversely, given a $Spin\left( 7\right)$-subbundle $Q\mathcal{M}$ of $F_+\mathcal{M}$ we can define an admissible four-form using equation \eqref{eq:qm8}. Therefore, we have found a one-to-one correspondence between positive four-forms $\Omega$ on $\mathcal{M}$ and $Spin\left( 7\right)$-structures on $\mathcal{M}$ inducing the given orientation on $\mathcal{M}$. 

A Riemannian metric on an oriented eight-dimensional manifold $\mathcal{M}$ implies a reduction of the structure group of the frame bundle from $GL\left(8,\mathbb{R}\right)$ to $SO(8)$. Since $Spin\left( 7\right)\subset SO(8)$, there is an associated metric $g$ to the $Spin\left( 7\right)$ structure on $\mathcal{M}$. We will call an admissible form $\Omega$ on $\mathcal{M}$ together with its associated metric $g$ a $Spin\left( 7\right)$-structure on $\mathcal{M}$, since, although $(\Omega,g)$ is not a $Spin\left( 7\right)$-structure it defines uniquely defines one. Let $\nabla$ be the Levi-Civita connection associated to $g$. We call $\nabla\Omega$ the torsion of $(\Omega,g)$ and we say that $(\Omega,g)$ is torsion free if $\nabla\Omega = 0$. The following proposition holds.

\begin{prop}
Let $\mathcal{M}$ be an oriented eight-dimensional manifold and let $(\Omega,g)$ be a $Spin\left( 7\right)$-structure on $\mathcal{M}$. Then the following are equivalent

\begin{itemize}

\item $(\Omega,g)$ is torsion-free.

\item $\mathrm{Hol}\left( g\right)\subseteq Spin\left( 7\right)$ and $\Omega$ is the induced four-form.

\item $\nabla\Omega =0$ on $\mathcal{M}$, where $\nabla$ is the Levi-Civita connection of $g$.

\item $d\Omega = 0$ on $\mathcal{M}$.

\end{itemize}

\end{prop}

\begin{proof}See lemma 11.5 in \cite{Salamon}.\end{proof} 

\noindent
Although $d\Omega = 0$ is a linear condition on $\Omega$, its restriction to $\Gamma\left(\mathcal{A}^4 \mathcal{M}\right)$ is non-linear.

\noindent 
We finish displaying cases in which a torsion-free $Spin(7)$ structure arises, for instance, from  from Calabi Yau manifolds of complex dimension two, three, four and from-torsion free $G_2$ structures. In particular, proposition 
\ref{prop:G2T2}
is the one we generalized in theorem \ref{thm:Spin(7)fromG2}, in the case with torsion. 
\begin{prop}
\label{prop:S7R4}
Suppose $\left(Y,g_{Y}\right)$ is a Riemannian four-dimensional manifold with holonomy $SU(2)$. Then $Y$ admits a complex structure form $J$, a K\"ahler form $\omega$ and a holomorphic volume form $\mathcal{V}$ such that $d\omega = d\mathcal{V} = 0$. Let $\mathbb{R}^{4}$ have coordinates $(x_{1},\dots,x_{4})$ and euclidean metric $h = dx^{2}_{1} + \cdots + dx^{2}_{4}$. Define a metric $g$ and a four-form $\Omega$ on $\mathbb{R}^{4}\times Y$ by $g=h\times g_{Y}$ and
\begin{equation}
\Omega = dx_{1234}+ (dx_{12} + dx_{34})\wedge\omega + (dx_{13}-dx_{24})\wedge {\rm Re} \mathcal{V} - (dx_{14} + dx_{23})\wedge {\rm Im} \mathcal{V} + \frac{1}{2}\omega\wedge\omega\, .
\end{equation}

\noindent
Then $(\Omega, g)$ is a torsion-free $Spin(7)$ structure on $\mathbb{R}^{4}\times Y$.
\end{prop}

\begin{proof}See proposition 13.1.1 in reference \cite{Joyce2007}.\end{proof} 

\noindent
In the previous proposition it possible to substitute $\mathbb{R}^4$ by $T^4$ and obtain a similar result.

\begin{prop}
Suppose $\left(Y,g_{Y}\right)$ is a Riemannian six-dimensional manifold with holonomy $SU(3)$. Let $\omega$ be the associated K\"ahler form and $\mathcal{V}$ the holomorphic volume form. Let $\mathbb{R}^2$ have coordinates $\left( x_{1}, x_{2}\right)$.  Define a metric $g$ and a four-form $\Omega$ on $\mathcal{M} = \mathbb{R}^2\times Y$ by $g = \left( dx^2_1+dx^2_2\right)\times g_{Y}$ and
\begin{equation}
\Omega = dx_{1}\wedge dx_{2}\wedge\omega + dx_{1}\wedge \Re{\rm e} \mathcal{V} - dx_{2}\wedge\Im{\rm m}\mathcal{V}+\frac{1}{2}\omega\wedge\omega\, .
\end{equation}

\noindent
Then $\left(\Omega, g\right)$ is a torsion-free $Spin(7)$-structure on $\mathcal{M}$.

\end{prop}

\begin{proof}See proposition 13.1.2 in reference \cite{Joyce2007}.\end{proof} 

\noindent
In the previous proposition it possible to substitute $\mathbb{R}^2$ by $T^2$ and obtain a similar result.

\begin{prop}
\label{prop:G2T2}
Suppose $\left(Y,g_{Y}\right)$ is a Riemannian seven-dimensional manifold with holonomy $G_{2}$. Let $\phi$ and $\ast\phi$ be the associated three-form and four-form. We define a metric $g$ and a four-form $\Omega$ on $\mathcal{M} =  \mathbb{R}\times Y$ by $g = dx^2_1\times g_{Y}$ and 
\beq
\Omega = dx\wedge\phi + \ast\phi
\eeq
where $x$ is the coordinate on $\mathbb{R}$. Then $\left(\Omega , g\right)$ is a torsion-free $Spin\left( 7\right)$-structure on $\mathcal{M}$.
\end{prop}
\begin{proof}See \cite[Prop. 13.1.3]{Joyce2007}.\end{proof} 

\noindent
In the previous proposition it possible to substitute $\mathbb{R}$ by $S^1$ and obtain a similar result.

\begin{prop}
\label{prop:SU4S7}
Suppose $\left(Y,g_{Y}\right)$ is a Riemannian eight-dimensional manifold with holonomy $SU(4)$, $Sp(2)$ or $SU(2)\times SU(2)$ and associated K\"ahler form $\omega$ and holomorphic volume form $\mathcal{V}$. Define a four-form $\Omega$ on $Y$ by
\beq
\Omega = \frac{1}{2}\omega\wedge\omega +{\rm Re} \mathcal{V} \ .
\eeq
Then $\left(\Omega , g\right)$ is a torsion-free $Spin\left( 7\right)$-structure on $Y$.
\end{prop}
\begin{proof}See \cite[Prop. 13.1.4]{Joyce2007}.\end{proof} 


\renewcommand{\leftmark}{\MakeUppercase{Bibliography}}
\phantomsection
\bibliographystyle{JHEP}
\bibliography{C:/Users/cshabazi/Dropbox/Referencias/References}
\label{biblio}
\end{document}